\documentclass[11pt,letterpaper]{article}
\usepackage{graphicx}
\usepackage{color,ifpdf,latexsym,url}
\usepackage[margin=1in]{geometry}
\usepackage{amsmath, amssymb, amsthm, amsfonts,alltt,clrscode}
\usepackage{array,colortbl,hhline,xcolor}
\usepackage{subcaption}

\newcommand{\onerush}{\texorpdfstring{$1\times1$}{1x1} Rush Hour}
\newcommand{\PSPACE}{PSPACE}
\newcommand{\NPSPACE}{NPSPACE}

\title{\onerush{} with Fixed Blocks is \PSPACE-complete}
\author{%
  Josh Brunner%
    \thanks{MIT Computer Science and Artificial Intelligence Laboratory,
      32 Vassar St., Cambridge, MA 02139, USA.
      \protect\url{{brunnerj,ikdc,edemaine,dylanhen,achester,avizeff}@mit.edu}}
\and
  Lily Chung\footnotemark[1]
\and
  Erik D. Demaine\footnotemark[1]
\and
  Dylan Hendrickson\footnotemark[1]
\and
  Adam Hesterberg\footnotemark[1]
\and
  Adam Suhl%
    \thanks{Algorand, Boston, MA, USA}
\and
  Avi Zeff\footnotemark[1]
}
\date{}

\newif\ifabstract
\abstracttrue
\newif\iffull
\ifabstract \fullfalse \else \fulltrue \fi

\usepackage{hyperref}
\hypersetup{breaklinks,bookmarksnumbered,bookmarksopen,bookmarksopenlevel=2}

{\makeatletter \hypersetup{pdftitle={\@title}}}

{\makeatletter
 \gdef\xxxmark{%
   \expandafter\ifx\csname @mpargs\endcsname\relax 
     \expandafter\ifx\csname @captype\endcsname\relax 
       \marginpar{xxx}
     \else
       xxx 
     \fi
   \else
     xxx 
   \fi}
 \gdef\xxx{\@ifnextchar[\xxx@lab\xxx@nolab}
 \long\gdef\xxx@lab[#1]#2{\textbf{[\xxxmark #2 ---{\sc #1}]}}
 \long\gdef\xxx@nolab#1{\textbf{[\xxxmark #1]}}
}

{\makeatletter \gdef\fps@figure{!htbp}}

\let\realbfseries=\bfseries
\def\bfseries{\realbfseries\boldmath}

\newtheorem{theorem}{Theorem}
\newtheorem{lemma}[theorem]{Lemma}

\theoremstyle{definition}
\newtheorem{definition}[theorem]{Definition}

\numberwithin{theorem}{section}

\let\epsilon=\varepsilon
\def\emph#1{\textbf{\textit{\boldmath #1}}}

\graphicspath{{figures/}}

\begin{document}
\maketitle


\begin{abstract}
  Consider $n^2-1$ unit-square blocks in an $n \times n$ square board, where
  each block is labeled as movable horizontally (only), movable vertically
  (only), or immovable --- a variation of Rush Hour with only $1 \times 1$
  cars and fixed blocks.
  We prove that it is \PSPACE-complete to decide whether a given block can
  reach the left edge of the board, by reduction from Nondeterministic
  Constraint Logic via 2-color oriented Subway Shuffle.
  By contrast, polynomial-time algorithms are known for deciding whether a
  given block can be moved by one space, or when each block either is
  immovable or can move both horizontally and vertically.
  Our result answers a 15-year-old open problem by Tromp and Cilibrasi,
  and strengthens previous \PSPACE-completeness results for
  Rush Hour with vertical $1 \times 2$ and horizontal $2 \times 1$ movable blocks and 4-color Subway Shuffle.
\end{abstract}

\section{Introduction}

In a \emph{sliding block puzzle}, the player moves blocks (typically
rectangles) within a box (often a rectangle) to achieve a desired
configuration.  Such puzzles date back to the 15 Puzzle,
invented by Noyes Chapman in 1874 and popularized
by Sam Loyd in 1891 \cite{Slocum-Sonneveld-2006},
where the blocks are unit squares.
One of the first puzzles to use rectangular pieces is the Pennant Puzzle
by L. W. Hardy in 1909, popularized under the name Dad's Puzzle from 1926,
whose 10 pieces require a whopping 59 moves to solve \cite{Gardner}.
In general, such puzzles are \PSPACE-complete to solve, even for
$1 \times 2$
blocks in a square box \cite{HD05,HD}, which was the original application
for the hardness framework Nondeterministic Constraint Logic (NCL).
For unit-square pieces (as in the 15 Puzzle), such puzzles can be solved
in polynomial time, though finding a shortest solution is NP-complete
\cite{Ratner-Warmuth-1990,FifteenPuzzle_TCS}.


In the 1970s, two famous puzzle designers --- Don Rubin in the USA
and Nobuyuki ``Nob'' Yoshigahara (1936--2004) in Japan --- independently
invented \cite{Rubin-v-Apple} a new type of sliding block puzzle,
where each block can move only horizontally or only vertically.
The motivation is to imagine each block as a car that can drive forward
and reverse, but cannot turn; the goal is to get one car (yours) to ``escape''
by reaching a particular edge of the board.
The original forms --- Rubin's ``Parking Lot'' \cite{rubin-parking}
and Nob's ``Tokyo Parking'' \cite{nob-parking} ---
imagined a poor parking-lot attendant trying to extract a car.
Binary Arts (now ThinkFun) commercialized Nob's $6 \times 6$ puzzles
as \emph{Rush Hour} in 1996, where a driver named Joe is
``figuring things out on their way to the American Dream''
\cite{evolution}.
The physical game design led to a design patent \cite{design-patent}
and many variations by ThinkFun since \cite{evolution}.%
\footnote{Sadly, to our knowledge, Rush Hour the puzzle was not an inspiration
  for Rush Hour the 1998 buddy cop film starring Jackie Chan and Chris Tucker.}
Computer implementations of the game at one point led to a lawsuit against
Apple and an app developer \cite{Rubin-v-Apple}.


The complexity of Rush Hour was first analyzed by Flake and Baum in 2002
\cite{FB}.
They proved that the game is \PSPACE-complete with the original piece types ---
$1 \times 2$ and $1 \times 3$ cars, which can move only in their long direction
--- when the goal is to move one car to the edge of a square board.
In 2005, Tromp and Cilibrasi \cite{TC} strengthened this result to use just
$1 \times 2$ cars (which again can move only in their long direction),
using NCL.
Hearn and Demaine \cite{HD,Hearn} simplified this proof, and proved analogous
results for triangular Rush Hour, again using NCL.
In 2016, Solovey and Halperin \cite{SoloveyH16} proved that Rush Hour is also
PSPACE-complete with $2 \times 2$ cars and immovable $0 \times 0$ (point)
obstacles.%
\footnote{Solovey and Halperin \cite{SoloveyH16} state their result in terms
  of unit-square cars amidst polygonal obstacles, but crucially allow the
  cars to be shifted by half of the square unit.  Phrased as cars aligned on
  a unit grid, these cars are effectively $2 \times 2$.}
Unlike $1 \times 2$ cars, which have an obvious
direction of travel (the long direction), $2 \times 2$ cars need to have a
specified direction, horizontal or vertical.

Back in 2002, Hearn, Demaine, and Tromp \cite{HD05,TC}%
\footnote{The open problem was first stated in the ICALP 2002 version of
  \cite{HD05}, based on discussions with John Tromp, as mentioned in \cite{TC},
  which is cited in the journal version of \cite{HD05}.}
raised a curious open problem:
might $1 \times 1$ cars suffice for \PSPACE-completeness of Rush Hour?
Like $2 \times 2$ cars, each $1 \times 1$ car has a specified direction,
horizontal or vertical.
This \emph{\onerush} problem behaves fundamentally differently: deciding
whether a specified car can move at all is polynomial time \cite{HD05,HD,TC},
whereas the analogous questions for $1 \times 2$ or $2 \times 2$ Rush Hour
(or for $1 \times 2$ sliding blocks)
are PSPACE-complete \cite{HD05,HD}.
Tromp and Cilibrasi \cite{TC} exhaustively searched all {\onerush} puzzles of a
constant size, and found that the length of solutions grew rapidly,
suggesting exponential-length solutions;
for example, the hardest $6 \times 6$ puzzle requires 732 moves.
They also suggested a variant where some cars cannot move at all
(perhaps they ran out of gas?), which we call \emph{fixed blocks}
by analogy with pushing block puzzles \cite{Push2F_CCCG2002},%
%
\footnote{Tromp and Cilibrasi \cite{TC} refer to {\onerush}
  as ``Unit (Size) Rush Hour'' and the fixed-block variant as
  ``Walled Unit Rush Hour''.}
as potentially easier to prove hard.


In this paper, we settle the latter open problem by Tromp and Cilibrasi
\cite{TC} by proving that {\onerush} with fixed blocks is \PSPACE-complete.
This result is the culmination of many efforts to try to resolve this problem
since it was posed in 2005; see the Acknowledgments.

Our reduction starts from NCL, and reduces
through another related puzzle game, Subway Shuffle.
In his 2006 thesis, Hearn \cite{Hearn,HD} introduced this type of puzzle as a
generalization of {\onerush}, again to help prove it hard.
Subway Shuffle involves motion planning of colored tokens on a graph
with colored edges, where the player can repeatedly move a token from one
vertex along an incident edge of the same color to an empty vertex, and the
goal is to move a specified token to a specified vertex.
Despite the generalization to graphs and colored tracks, the complexity
remained open until 2015, when De Biasi and Ophelders \cite{DBO} proved it
\PSPACE-complete by a reduction from NCL.
Their proof works even when the graph is planar and uses just four colors.

We use a variant on Subway Shuffle where the graph is directed, and tokens can
travel only along forward edges.
In Section~\ref{sec:subway shuffle PSPACE}, we prove that directed
Subway Shuffle is \PSPACE-complete even with planar graphs and just two colors,
by a proof similar to that of De Biasi and Ophelders \cite{DBO}.
In Section~\ref{sec:rush hour PSPACE}, we then show that this construction
uses a limited enough set of vertices that it can actually be embedded
in the grid and simulated by {\onerush}, proving \PSPACE-completeness of
the latter with fixed blocks.
We conclude with open problems in Section~\ref{sec:conclusion}.

\section{Basics}

First we precisely define the problems introduced above.

\begin{definition}
  In \emph{Rush Hour}, we are given a square grid containing nonoverlapping \emph{cars}, which are rectangles with a specified orientation, either horizontal or vertical. A legal move is to move a car one square in either direction along its orientation, provided that it remains within the square and does not intersect another car. The goal is for a designated special car to reach the left edge of the board. We also allow \emph{fixed blocks}, which are spaces cars cannot occupy.
\end{definition}

\begin{definition}
  \emph{\onerush} is the special case of Rush Hour where each car is $1\times1$.
\end{definition} 

\begin{definition}
  In \emph{Subway Shuffle}, we are given a planar undirected graph where each edge is colored and some vertices contain a colored token. A legal move is to move a token across an edge of the same color to an empty vertex. The goal is for a designated special token to reach a designated target vertex.
\end{definition}

\begin{definition}
  In \emph{oriented Subway Shuffle}, we are given a planar directed graph where each edge is colored and some vertices contain a colored token. A legal move is to move a token across an edge of the same color, in the direction of the edge, to an empty vertex, and then flip the direction of the edge. The goal is for a designated special token to reach a designated target vertex.
\end{definition}

\begin{lemma}\label{lem:in PSPACE}
  Subway Shuffle, oriented Subway Shuffle, and Rush Hour are in \PSPACE.
\end{lemma}

\begin{proof}
  We can solve these problems in nondeterministic polynomial space by guessing each move, and accepting when the special car or token reaches its goal. So all three problems are contained in \NPSPACE, and by Savitch's theorem \cite{SAVITCH} they are in \PSPACE.
\end{proof}

\section{2-color Oriented Subway Shuffle is \PSPACE-complete} \label{sec:subway shuffle PSPACE}

In this section, we show that 2-color oriented Subway Shuffle is \PSPACE-complete. To do so, we reduce from nondeterministic constraint logic, which is \PSPACE-complete \cite{HD}. Our reduction is an adaption of the proof in \cite{DBO} in which the gadgets use only two colors (instead of four) and work in the oriented case.

We actually prove a slightly stronger result in Theorem~\ref{thm:ss valid}: that 2-color oriented Subway Shuffle is \PSPACE-complete even with a restricted vertex set, and with a single unoccupied vertex. A vertex is \emph{valid} if it has degree at most 3, and has at most 2 edges of a single color attached to it; these vertices are shown in Figure~\ref{fig:valid vertices}. Our proof of \PSPACE-hardness will only use valid vertices. 

\begin{figure}
  \centering
  \includegraphics[width=.6\linewidth]{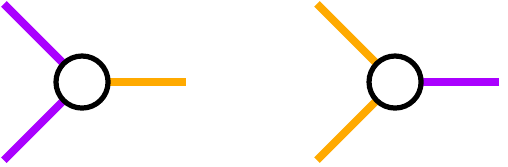}
  \caption{The valid Subway Shuffle vertices with degree 3. Every vertex with degree 1 or 2 is valid.}
  \label{fig:valid vertices}
\end{figure}

\begin{theorem}\label{thm:ss valid}
  2-color oriented Subway Shuffle with only valid vertices and exactly one unoccupied vertex is \PSPACE-complete.
\end{theorem}

\begin{proof}
  Containment in \PSPACE{} is given by Lemma~\ref{lem:in PSPACE}. To show hardness, we reduce from planar NCL with AND and protected OR vertices.
  
 In constraint logic, a \emph{protected OR vertex} is an OR vertex (one with three blue edges) such that two edges, due to global constraints, cannot simultaneously point towards the vertex. NCL is still \PSPACE-complete when every OR vertex is protected \cite{HD}. Because of this global constraint, there are only five possible states that a protected OR vertex can be in. In particular, there are only four possible transitions between the states of a protected OR vertex. Our OR gadget only allows these four transitions; in particular it does not allow the transition between only the leftmost edge pointing inward and only the rightmost edge pointing inward, which is the defining transition that a protected OR vertex does not have compared to a normal OR vertex.

  The Subway Shuffle instance we construct will have only a single empty vertex (other than the target vertex), called the \emph{bubble}, which moves around the graph opposite the motion of tokens. Our vertex and edge gadgets work by having the bubble enter them, move around a cycle, and then exit at the same vertex. The effect is that each edge in the cycle flips and each token in the cycle moves across one edge, except that one token by the entrance moves twice.

  The general structure of the reduction is as follows. First, we choose any rooted spanning tree on the dual graph of the constraint logic graph. This rooted spanning tree will determine the path the bubble takes to get from one vertex or edge to another. For each edge and vertex in the constraint logic graph, we will replace it with a subway shuffle gadget. The constraint logic edges which are part of our spanning tree will have a path for the bubble to cross through them. Each face of the CL graph has paths connecting vertex gadgets and edge gadgets as necessary to allow the bubble to visit each gadget.

  When playing the constructed Subway Shuffle instance, the bubble begins at the root of the spanning tree. The bubble can move down the tree by crossing edge gadgets until reaching a desired face. It then enters a vertex or edge gadget, goes around a cycle, and exits. A sequence of moves of this form corresponds to flipping a constraint logic edge or reconfiguring a vertex (that is, changing which constraint logic edge(s) are used to satisfy that vertex and therefore are locked from being flipped away from it). The bubble can always travel back up the spanning tree to the root, and from there visit any face and then any CL vertex or edge.

Now we will describe the various gadgets that implement constraint logic in Subway Shuffle. Many places in the gadget figures have an empty vertex attached to them; this represents where the gadget is connected to the spanning tree. Entering through these vertices is the only way the bubble can interact with a gadget.
  
  The edge gadget is shown in Figure~\ref{fig:Subway Shuffle edge}. The two vertices and edge at the bottom and top of the edge gadget (in a gray box in the figure) are shared with the connecting vertex gadget. The edge gadget consists of five interlocking cycles. The edge can be flipped by rotating each of the five cycles in order, as shown in Figure~\ref{fig:Subway Shuffle edge transitions}. The bubble rotates a cycle by entering at the appropriate white vertex, and then moving around the cycle, and finally exiting where it entered.

  If the edge is in the spanning tree, we include the rightmost vertex called the \emph{exit}, which allows the bubble to visit the edge gadget and pass through to face on the other side of the edge. We place the edge gadget in the orientation so that the entrance is on the face closer to the root of the spanning tree of the dual graph.

  There are two kinds of edges in constraint logic: red and blue edges. The only difference is that they have different weights for the constraint logic constraints. Blue edges are as shown in Figure~\ref{fig:Subway Shuffle edge} (and can be rotated); red edges are the same gadget, but reflected.
 
 Edge gadgets connect to vertex gadgets by sharing the two vertices and edge marked in a gray box. In the edge gadget, when the vertex colors and edge direction are as shown in the edge gadget figure, the edge is \emph{unlocked}, which means that the bubble is free to flip the direction of that edge. The vertex gadget's colors take precedence for the shared edges and vertices. When they do not match those shown in the edge gadget figure, we say the edge is \emph{locked}. When this happens, it becomes impossible for the bubble to rotate first cycle, and thus prevents the bubble from flipping the edge. This mechanism is what allows the gadgets to enforce the constraints of the vertices in the constraint logic graph. Edges are only ever locked while pointing into a vertex because all of the constraints in constraint logic only give lower bounds on the number of inward pointing edges. When an edge is pointing away from a vertex, some of the cycles in the vertex will be impossible to rotate, preventing the bubble from unlocking other edges.


  \begin{figure}
    \centering
    \begin{subfigure}{.2\linewidth}
      \centering
      \includegraphics[width=\linewidth]{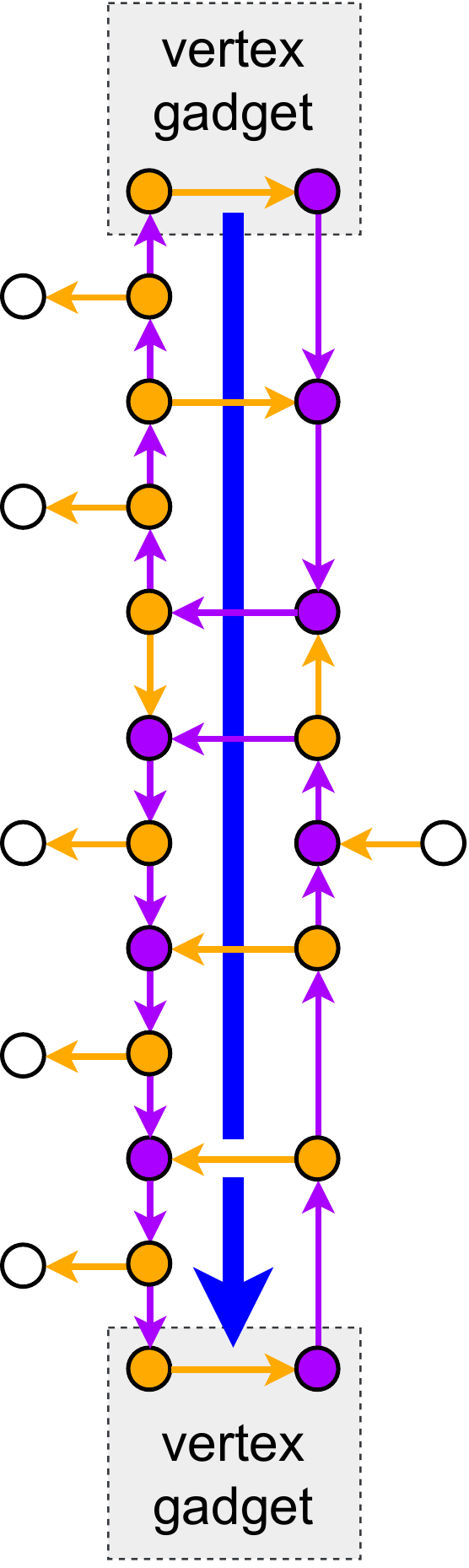}
      \caption{}
      \label{subfig:down unlocked}
    \end{subfigure}
    \hfil
    \begin{subfigure}{.2\linewidth}
      \centering
      \includegraphics[width=\linewidth]{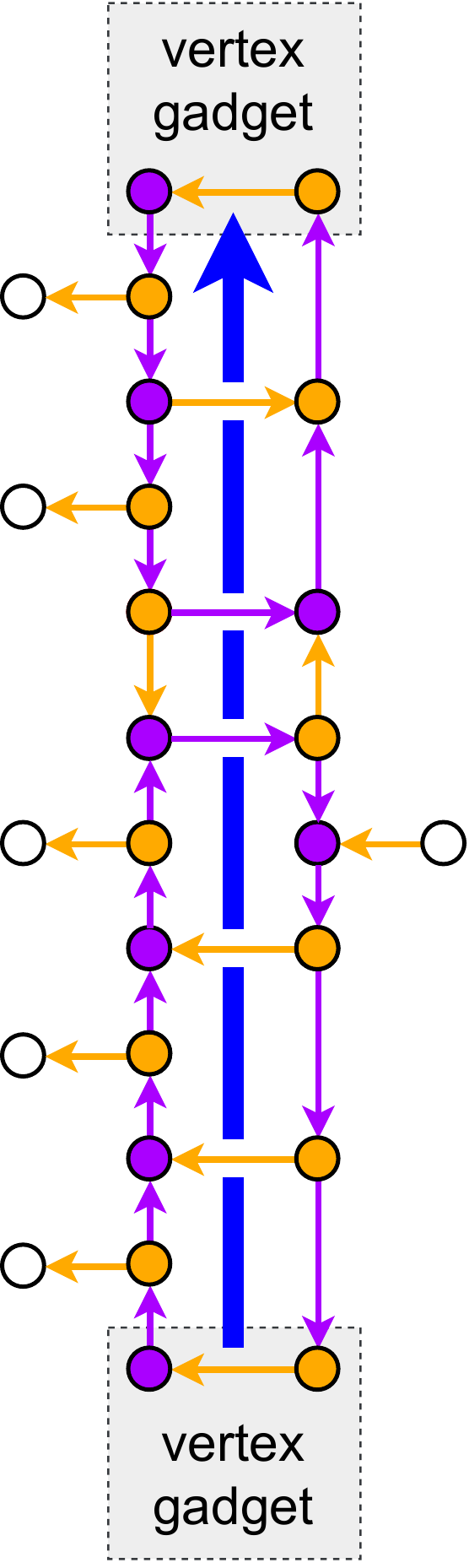}
      \caption{}
      \label{subfig:up unlocked}
    \end{subfigure}
    \hfil
    \begin{subfigure}{.2\linewidth}
      \centering
      \includegraphics[width=\linewidth]{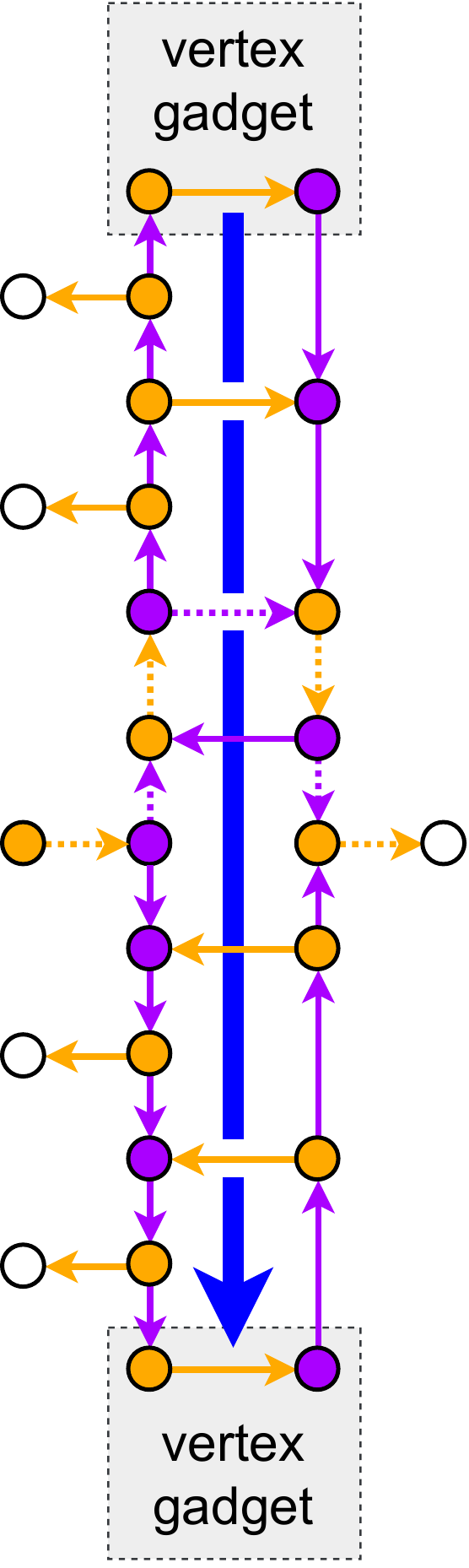}
      \caption{}
      \label{subfig:down traversed}
    \end{subfigure}
    \hfil
    \begin{subfigure}{.2\linewidth}
      \centering
      \includegraphics[width=\linewidth]{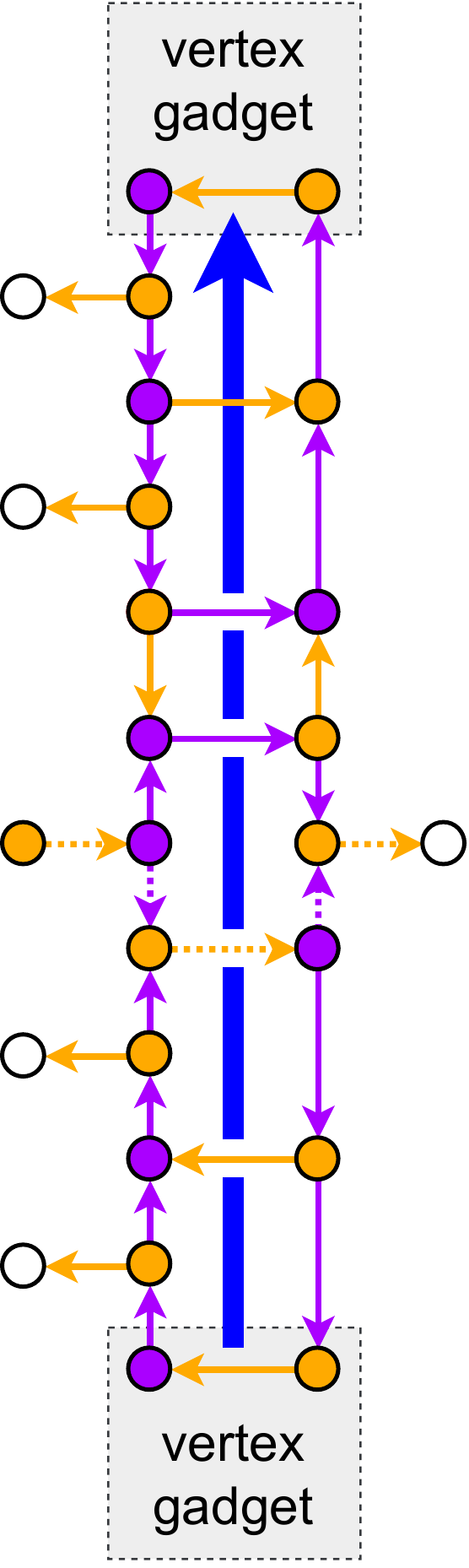}
      \caption{}
      \label{subfig:up traversed}
    \end{subfigure}
    \caption{The edge gadget for 2-color oriented Subway Shuffle, shown (a) directed down and unlocked, (b) directed up and unlocked, (c) directed down after the bubble has passed through, and (d) directed up after the bubble has passed through. This gadget is based on the edge gadget in \cite{DBO}.}
    \label{fig:Subway Shuffle edge}
  \end{figure}

  \begin{figure}
    \centering
    \begin{subfigure}{.18\linewidth}
      \centering
      \includegraphics[width=\linewidth]{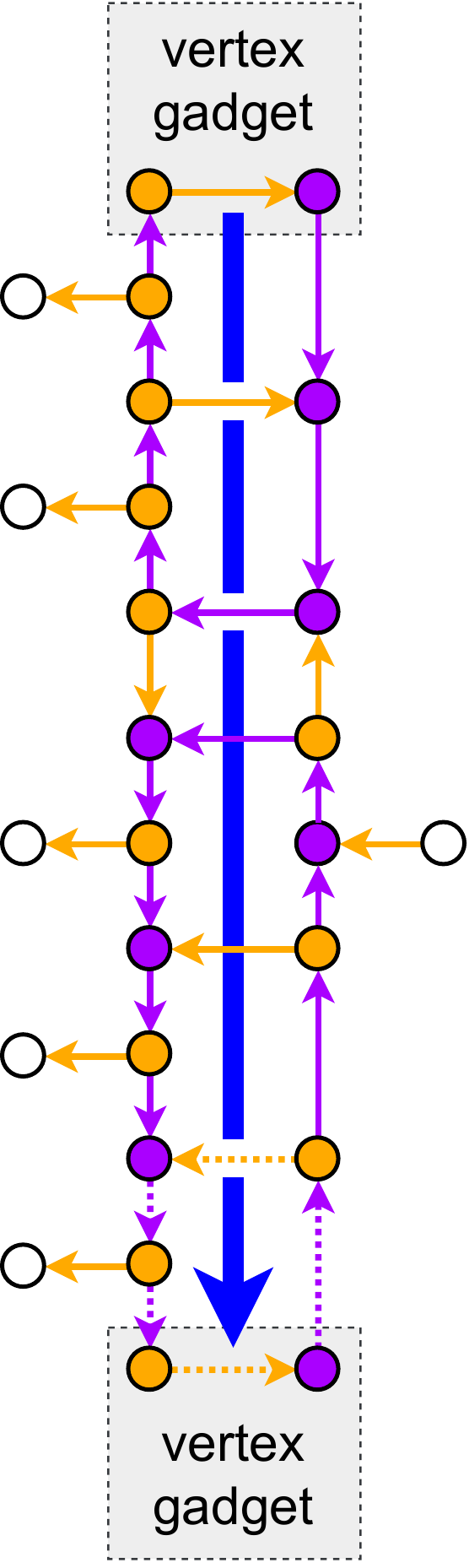}
      \caption{}
      \label{subfig:rotation1}
    \end{subfigure}
    \hfil
    \begin{subfigure}{.18\linewidth}
      \centering
      \includegraphics[width=\linewidth]{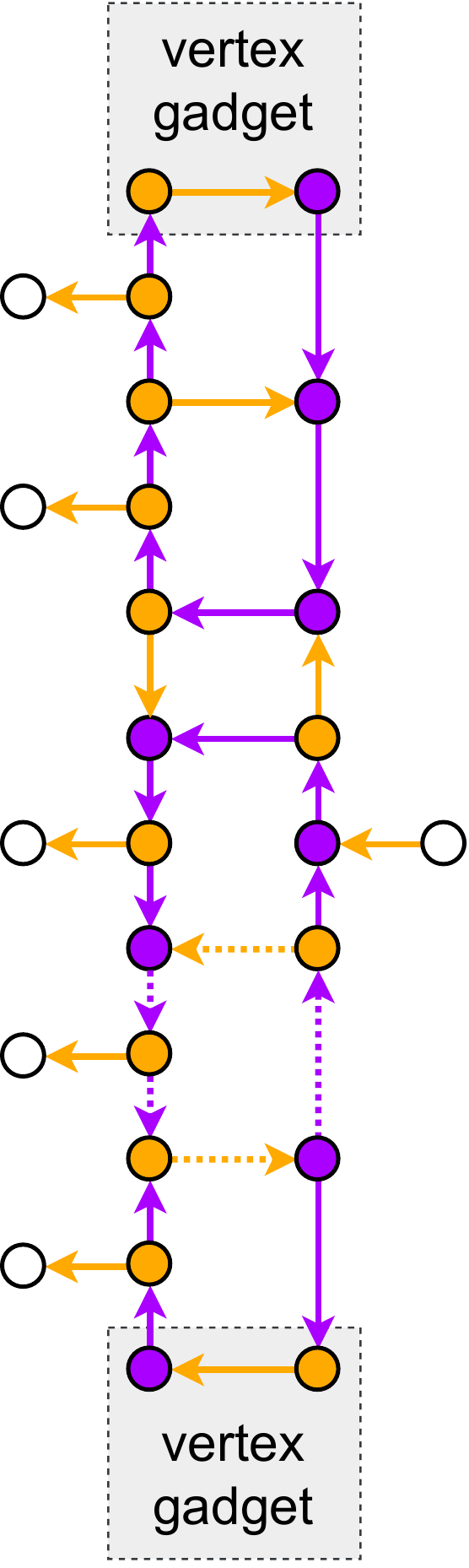}
      \caption{}
      \label{subfig:rotation2}
    \end{subfigure}
    \hfil
    \begin{subfigure}{.18\linewidth}
      \centering
      \includegraphics[width=\linewidth]{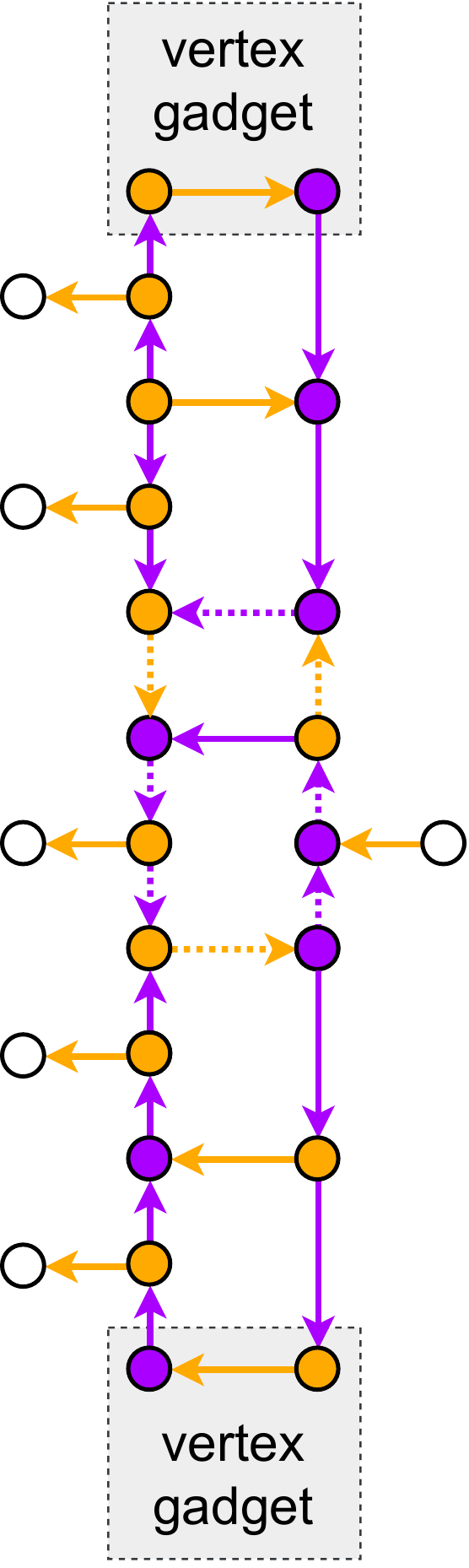}
      \caption{}
      \label{subfig:rotation3}
    \end{subfigure}
    \hfil
    \begin{subfigure}{.18\linewidth}
      \centering
      \includegraphics[width=\linewidth]{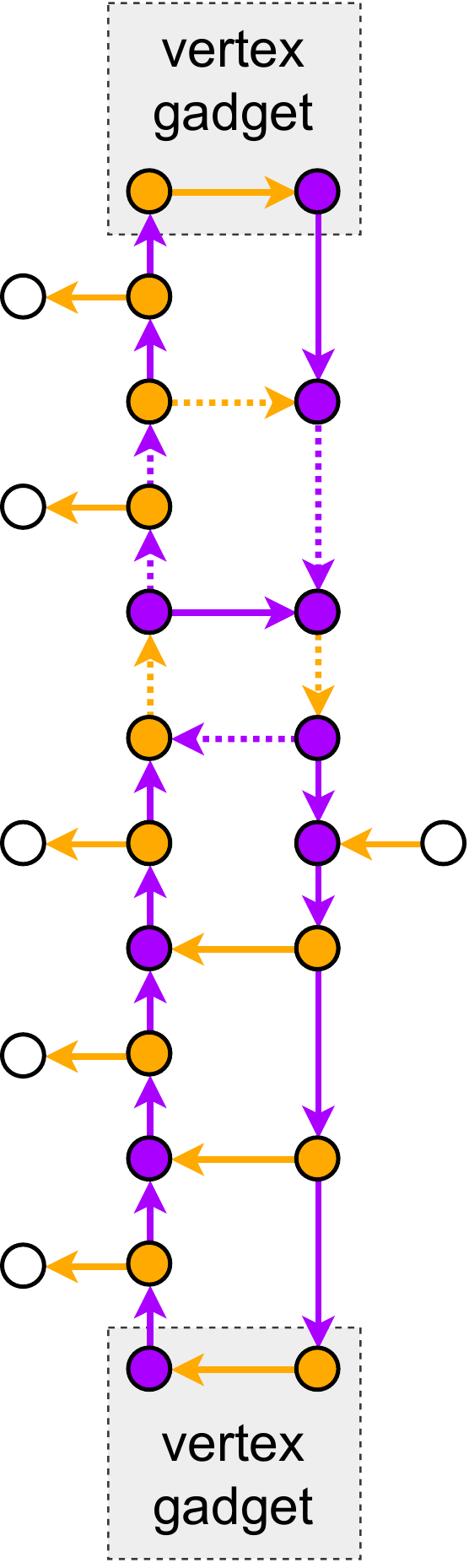}
      \caption{}
      \label{subfig:rotation4}
    \end{subfigure}
    \hfil
    \begin{subfigure}{.18\linewidth}
      \centering
      \includegraphics[width=\linewidth]{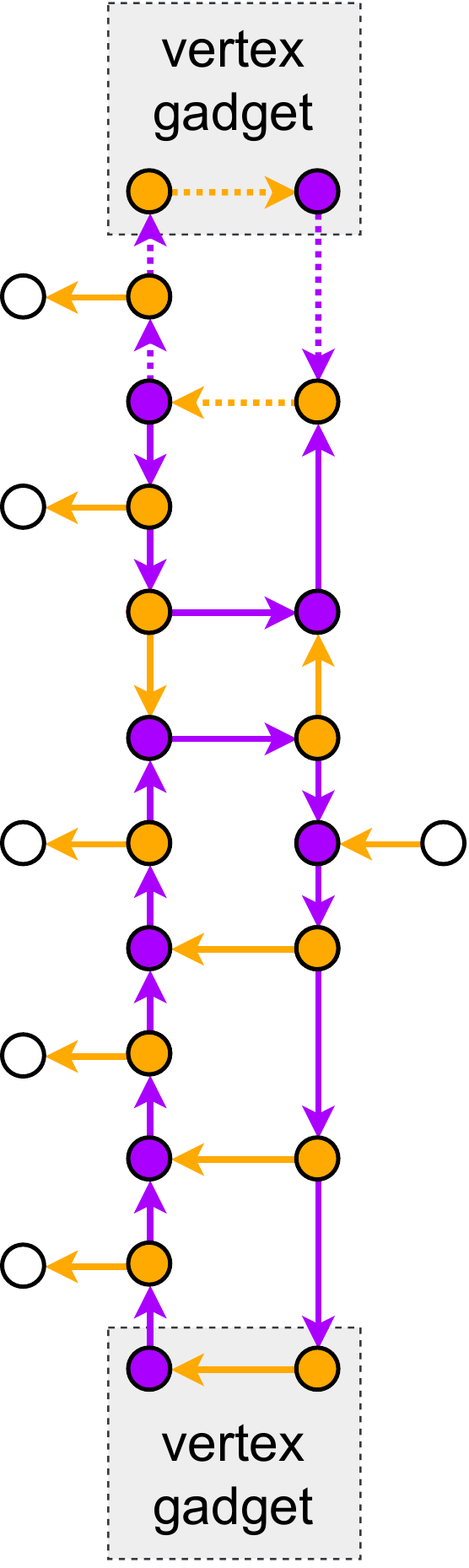}
      \caption{}
      \label{subfig:rotation5}
    \end{subfigure}
    \caption{The five cycles that the bubble rotates to flip the orientation of an edge gadget. For each cycle, the bubble enters at the white vertex, goes around the dotted cycle, and leaves where it entered. Note that the colors and orientation of the edge and vertices that connect to the vertex gadget will not always match what is shown in this figure. When they do not, we say the edge is \emph{locked} by the corresponding vertex gadget, and it is not possible to rotate the dotted cycle. }
    \label{fig:Subway Shuffle edge transitions}
  \end{figure}

  The AND vertex gadget is shown in Figure~\ref{fig:Subway Shuffle and}. Whenever the bubble is not visiting the vertex gadget, either the blue (weight 2) edge or both red (weight one) edges are locked to point towards the vertex. If all three edges are pointing towards the vertex, the bubble can visit the vertex gadget (at the top entrance) and go around the cycle to switch which edges are locked. This implements the constraints on a NCL AND vertex.

  \begin{figure}
    \centering
    \begin{subfigure}{.4\linewidth}
      \centering
      \includegraphics[width=\linewidth]{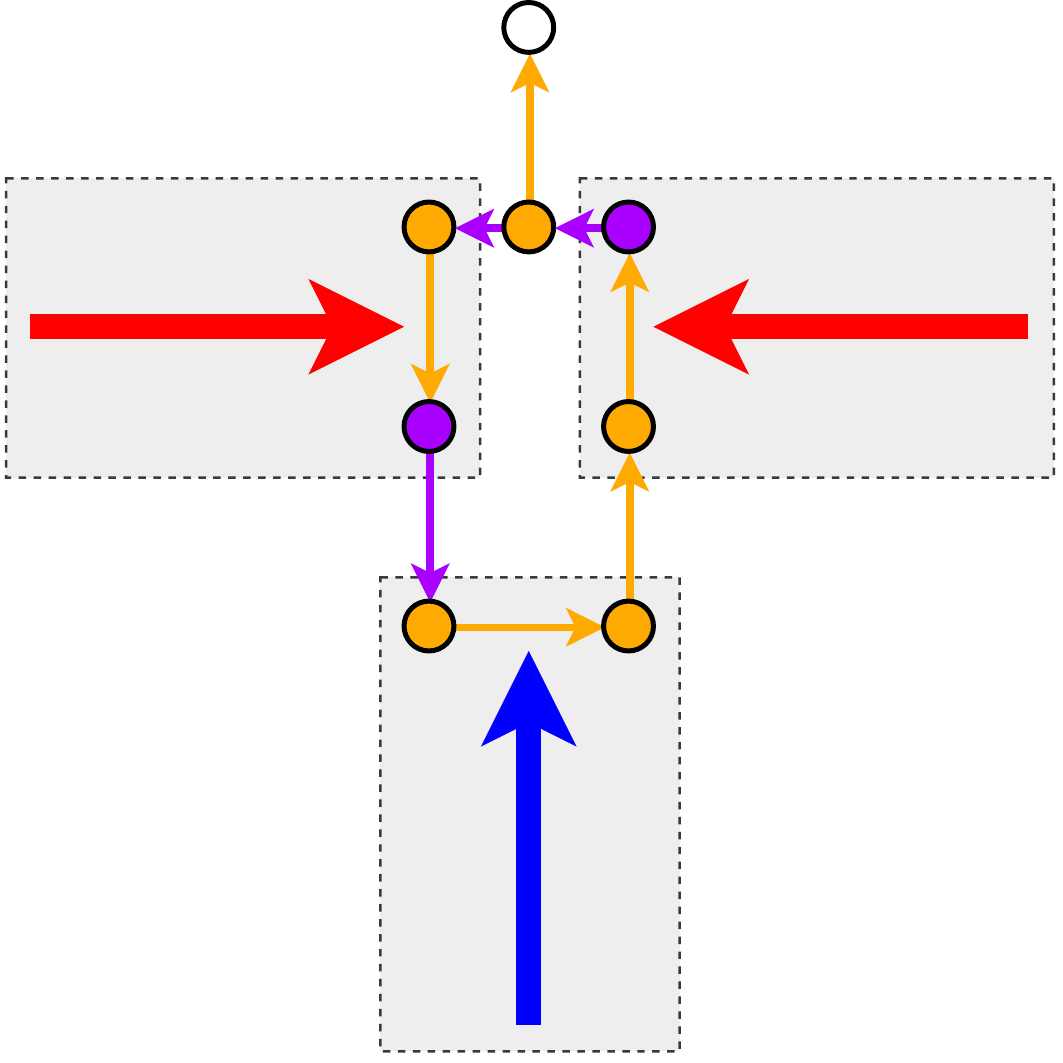}
      \caption{All edges oriented in, with the blue edge locked.}
      \label{subfig:and 1}
    \end{subfigure}
    \hfil
    \begin{subfigure}{.4\linewidth}
      \centering
      \includegraphics[width=\linewidth]{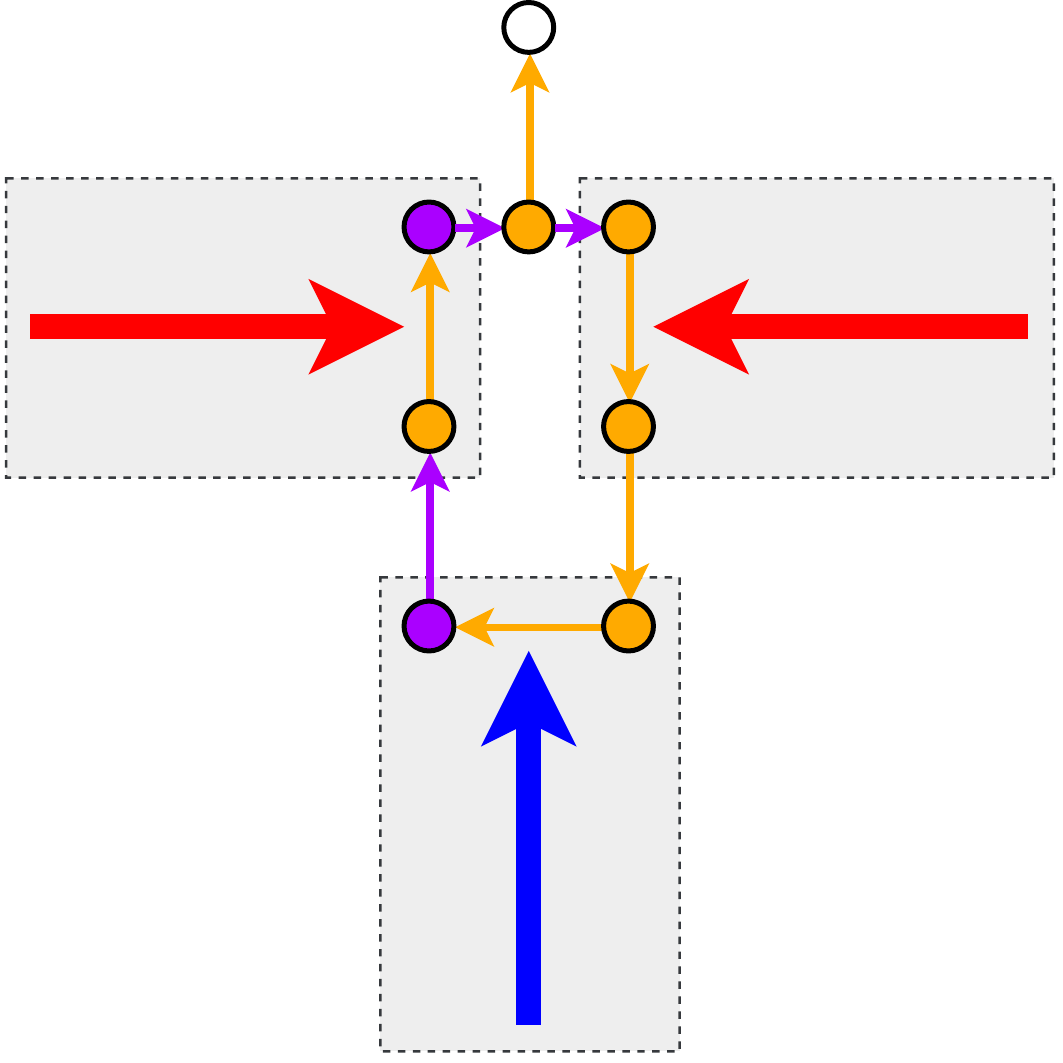}
      \caption{All edges oriented in, with both red edges locked.}
      \label{subfig:and 2}
    \end{subfigure}
    \caption{The AND vertex gadget for 2-color oriented Subway Shuffle. This gadget is based on the AND vertex gadget in \cite{DBO}.}
    \label{fig:Subway Shuffle and}
  \end{figure}



  Our protected OR vertex gadget is shown in Figure~\ref{fig:Subway Shuffle or transitions}. The two protected edges are the leftmost and rightmost edges, so we can assume that they never both point towards the vertex. The gadget has three entrances. The gadget can be in five possible states corresponding to the five possible states of a CL protected OR. In each state, the edges which are locked in correspond to the set of edges that are pointing inward in the corresponding state of a CL protected OR. In the first state, the left edge is locked, and the other two are free. In the second state, the middle edge is also locked. In the third state, only the middle edge is locked. In the fourth state, both the middle and right edges are locked. Finally, in the fifth state, only the right edge is locked. To get from one state to the next, the bubble rotates a single cycle. The fives states and the transitions between them are shown in Figure~\ref{fig:Subway Shuffle or transitions}. The only transitions between states are to the next and previous states. To transition from one state to the next, the bubble goes around the cycle indicated by the dotted edges.

  \begin{figure}
    \centering
    \begin{subfigure}{.49\linewidth}\centering
      \includegraphics[width=.95\linewidth]{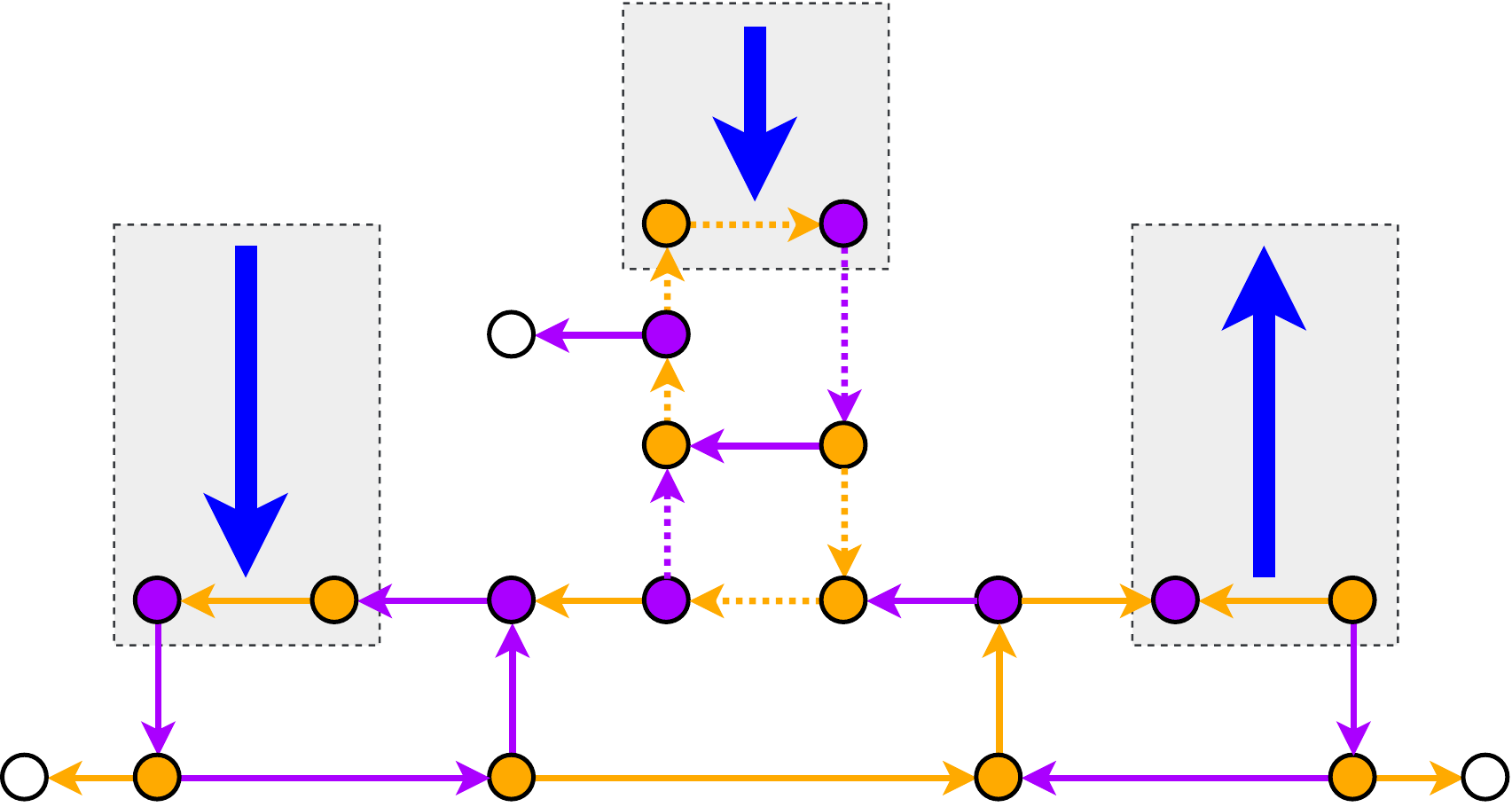}
      \caption{State 1: The left edge is locked. The middle edge is unlocked and pointing in. The right edge is pointing out.}
    \end{subfigure}\hfill
    \begin{subfigure}{.49\linewidth}\centering
      \includegraphics[width=.95\linewidth]{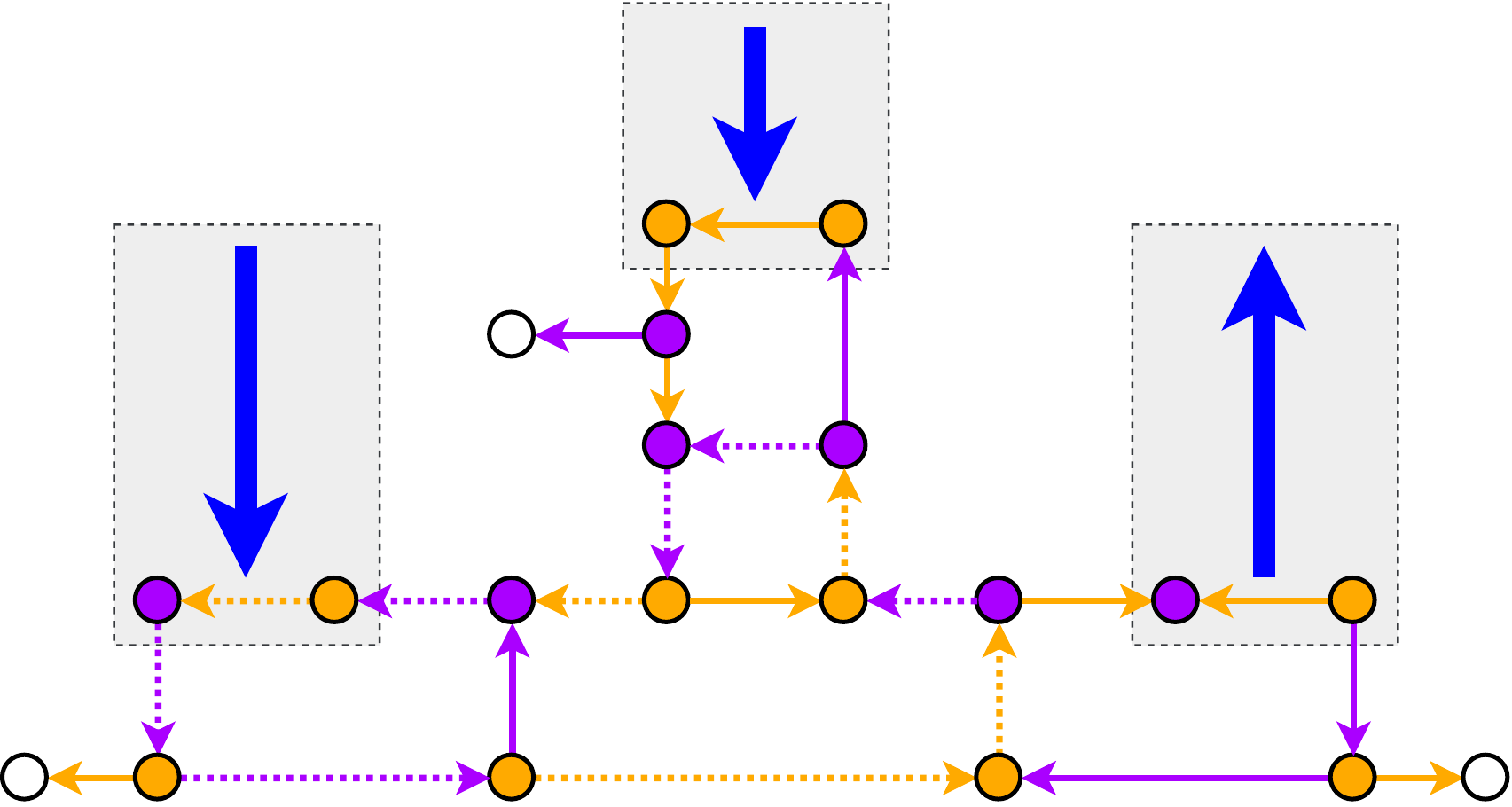}
      \caption{State 2: The left and middle edges are locked. The right edge is pointing out.\\}
    \end{subfigure}
    \begin{subfigure}{.49\linewidth}\centering
      \includegraphics[width=.95\linewidth]{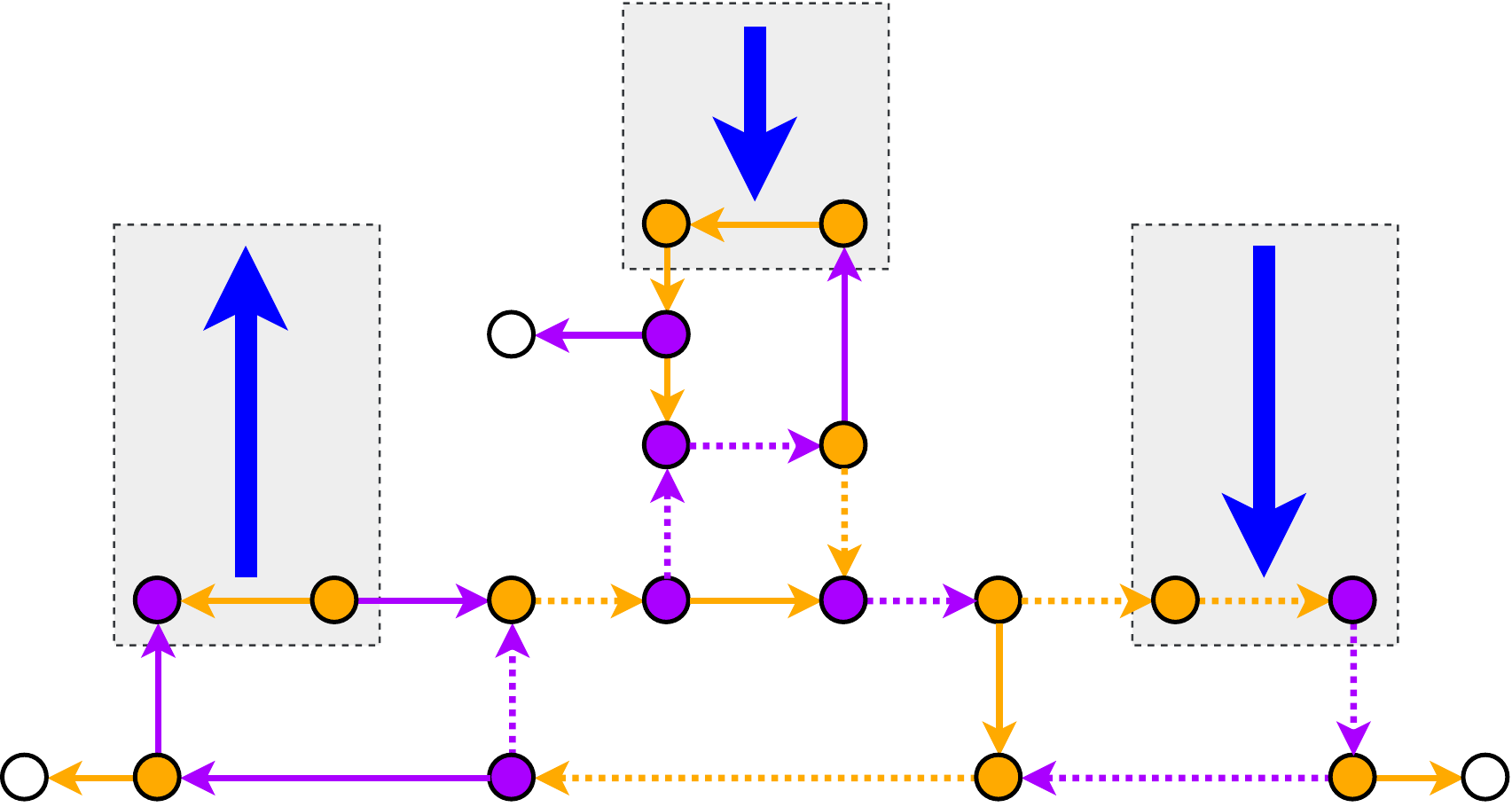}
      \caption{State 3: The left edge is pointing out. The middle edge is locked. The right edge is unlocked and pointing in.}
    \end{subfigure}\hfill
    \begin{subfigure}{.49\linewidth}\centering
      \includegraphics[width=.95\linewidth]{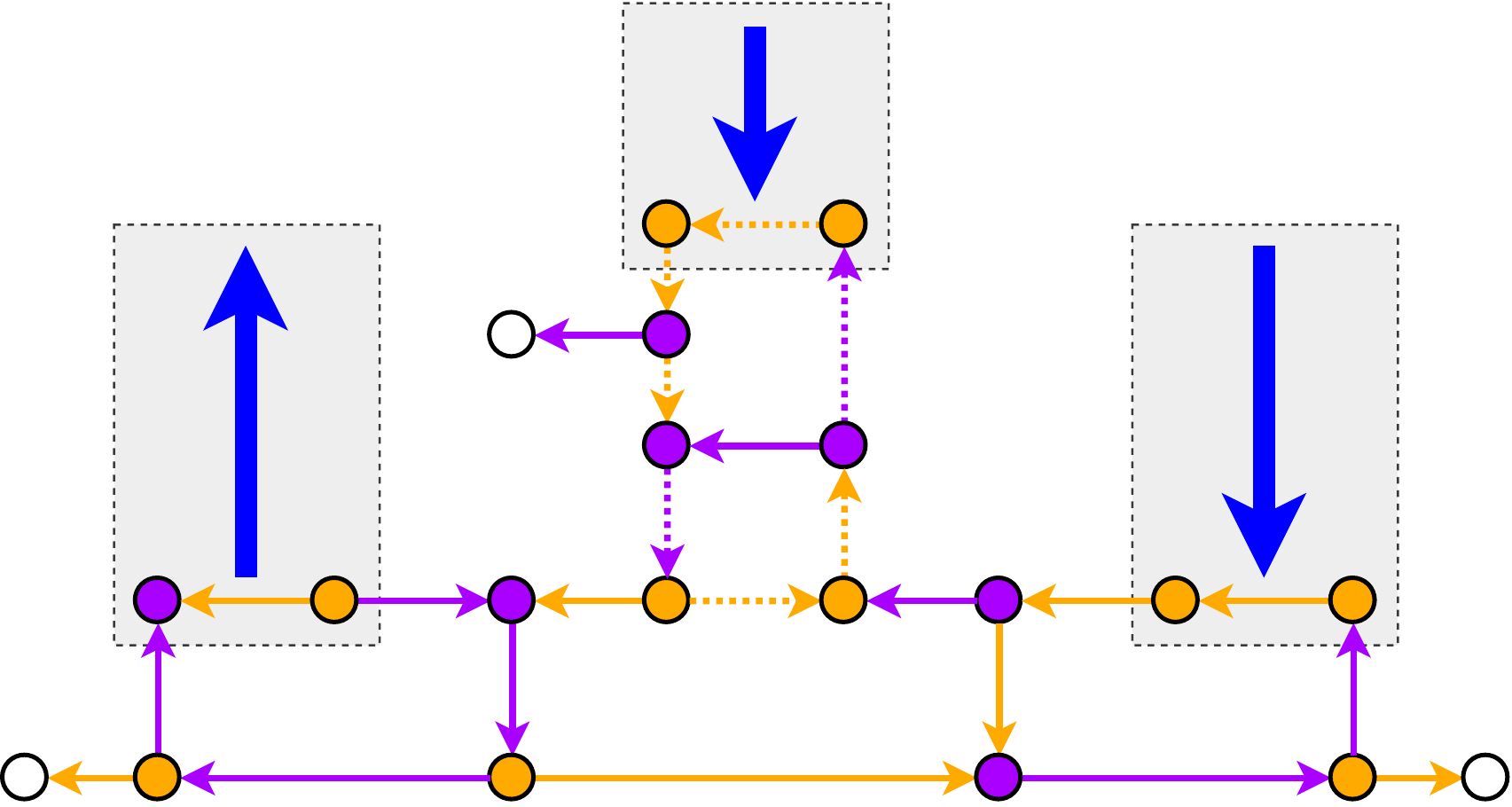}
      \caption{State 4: The left edge is pointing out. The middle and right edges are locked.\\}
    \end{subfigure}
    \begin{subfigure}{.49\linewidth}\centering
      \includegraphics[width=.95\linewidth]{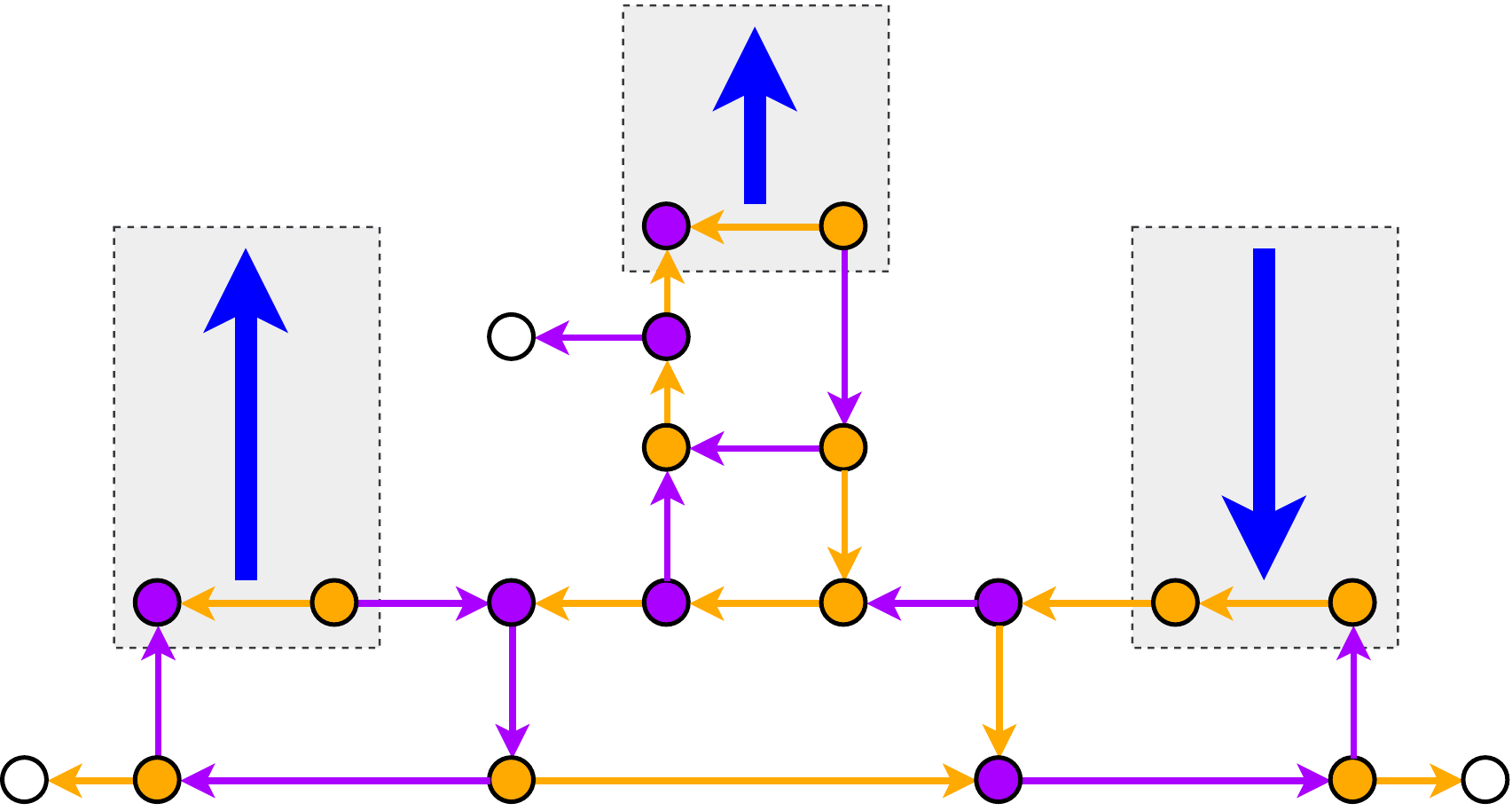}
      \caption{State 5: The left and middle edges are pointing out. The right edge is locked.}
    \end{subfigure}
    \caption{The five states of the protected OR vertex. The dotted edges show the cycle that is rotated to transition to the next state. Note that each state is defined only by which edged are locked in; the other unlocked edges can be either pointing in or out in each state.}
    \label{fig:Subway Shuffle or transitions}
    \end{figure}

  Our last gadget is the win gadget, shown in Figure~\ref{fig:Subway Shuffle win}. It is placed attached to the edge gadget corresponding to the target edge in the constraint logic instance, and allows the player to win the Subway Shuffle instance when that edge can be flipped. 

  In the first state shown, the target edge is pointing away. If the bubble arrives at the win gadget, it cannot accomplish anything. If the target edge is flipped so it now points toward the win gadget, we will be in the second state. Then the bubble can enter the win gadget at the top entrance and go around the indicated cycle, moving the special token one to the left. Finally, the bubble can enter at the bottom entrance to move the special token across to the target vertex.

  \begin{figure}
    \centering
    \begin{subfigure}{.4\linewidth}
      \includegraphics[width=\linewidth]{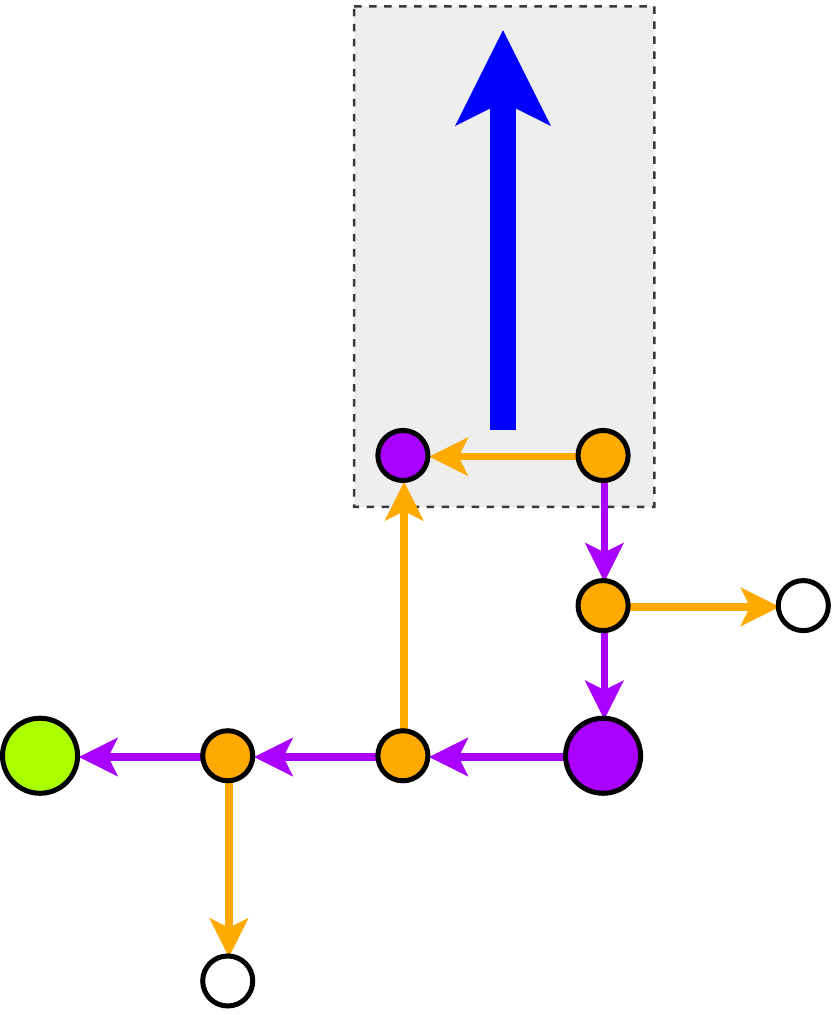}
      \caption{The locked win gadget. The bubble cannot do anything here.\\}
    \end{subfigure}
    \hfil
    \begin{subfigure}{.4\linewidth}
      \includegraphics[width=\linewidth]{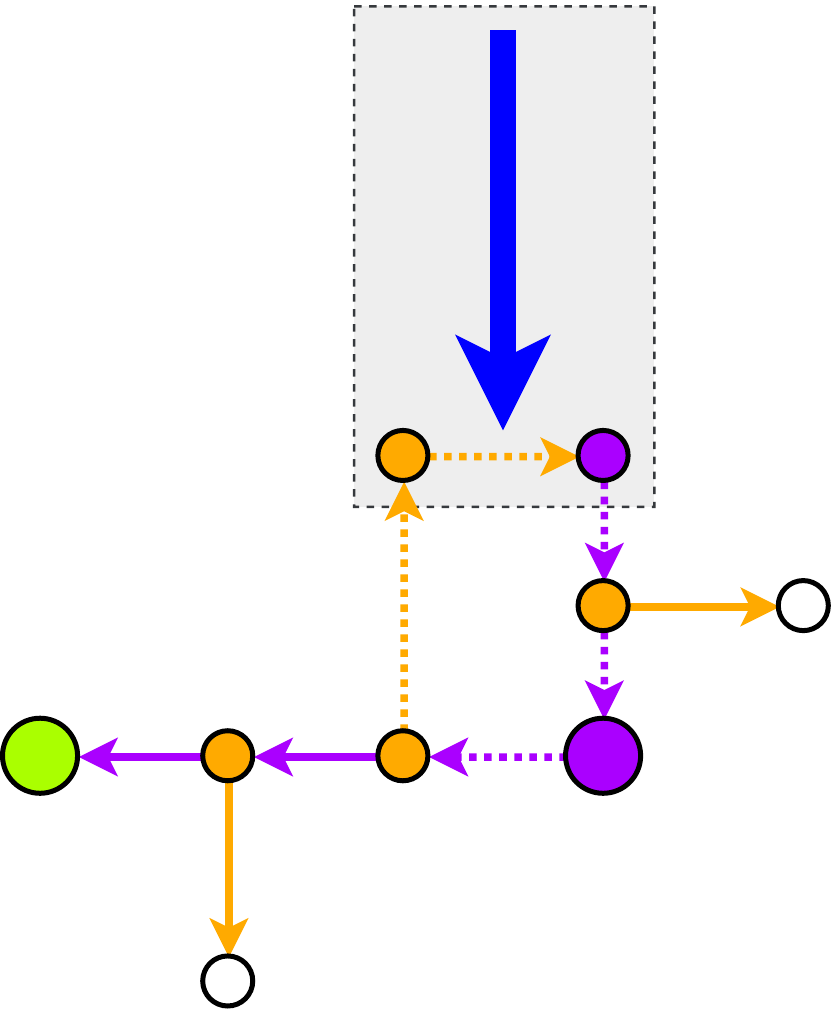}
      \caption{The unlocked win gadget. Now that the blue edge is pointing into the gadget, the indicated cycle can be rotated.}
    \end{subfigure}
    \caption{The win gadget for 2-color oriented Subway Shuffle. The target edge starts pointing away. If it is flipped, the bubble can enter the win gadget once at each entrance to move the (bottom right) purple special token to the (bottom left, green) target vertex. This gadget is based on the FINAL gadget in \cite{DBO}.}
    \label{fig:Subway Shuffle win}
    \end{figure}

  To allow the bubble to reach every gadget, we connect the entrances and exits of gadgets which are on the same face of the CL graph. This simply requires a tree connecting these vertices for each face. Each face other than the root of the spanning tree has exactly one edge exit on it; we orient the edges on that face to point towards this exit. The color of these edges does not matter, provided all vertices are valid and the token at the tail of an edge is the same color. For the face which is the root of the spanning tree, the tree connecting entrances has one vertex without a token, and the edges point towards it; this is where the bubble starts.

  Now we show how the gadgets prevent any moves other than the moves outlined above that simulate the NCL instance. First we consider the edge gadget.
  It is easy to check that while rotating any of the dotted cycles in an edge gadget, there are only two legal moves other than continuing the cycle. The first one is leaving through the exit vertex during the third cycle. This is equivalent to the bubble just using the throughway in the edge gadget to reach the rest of the spanning tree after turning only the first two cycles. By Lemma~\ref{partial edges}, this is never useful. The other legal move is while turning the first, fourth, or fifth cycle, it is possible for the bubble to move into the connecting vertex gadget through the shared vertices. We will show that nothing useful can be accomplished here when we consider the vertex gadgets. Similarly, it will also be possible for the bubble to come from a vertex and enter the edge gadget through the shared vertices. We show this is not useful in Lemma~\ref{no edge leakage}.
  
  \begin{lemma}
    \label{no edge leakage}
    It is never useful for the bubble to enter an edge gadget directly from a vertex gadget through the shared vertices.
  \end{lemma}
  \begin{proof}
    We need to check the up, down, up traversed, and down traversed configurations.
    
    In most configurations, there are no legal moves to enter the edge gadget from the shared vertices. The only configuration where this is possible is from the orange token on the top left of the upward pointing edge gadget. From here, it can move through a path of three tokens before it gets stuck. At that point, the only legal move is to undo the last three moves and exit the same way it entered.
    \end{proof}
  \begin{lemma}
    \label{partial edges}
    It is never useful to turn some of the cycles in an edge gadget without turning all of them.
  \end{lemma}
  \begin{proof}
    If you turn some of the cycles, but not all of them, then both ends of the edge gadget will be in the pointing outward configuration. For all of the vertex gadgets, there are no transitions that require the outward pointing configuration, so the edge gadget being in this configuration never lets you make a move that you could not make if you finished turning all of the cycles in an edge gadget.

    We also need to make sure that turning only some cycles, and then entering an edge gadget from a vertex gadget (as in Lemma~\ref{no edge leakage}), does not allow you to do anything. If we look at all of the partial edge configurations as shown in Figure~\ref{fig:Subway Shuffle edge transitions}, there is no way to access anything from any of these configurations. We also need to check the configurations that arise from partially rotating an edge and then traversing it. Since it is not possible to reach the traverse paths from entering from a vertex gadget, these configurations also do not let the bubble do anything else useful.
  \end{proof}

  Now we consider the AND gadget. Since the entire gadget is a single cycle, there is nothing the bubble can do within the gadget while turning the cycle. While turning the cycle, the bubble can try to enter an edge gadget through one of the shared vertices; however, we have already shown that the this is never useful in Lemma~\ref{no edge leakage}.

  We also need to consider if the bubble enters the vertex gadget from an edge on one of the shared vertices. It will never be able to move around the entire cycle because the orange vertex at the top will not be accessible. The only other thing the bubble can do is try to enter a different edge gadget, but we already showed this is not useful in Lemma~\ref{no edge leakage}.
  
  Now we consider the OR gadget. First we look at each of the four cycles. While turning the first cycle, the only legal move that is not continuing the cycle is moving the purple token just to the right of the cycle. However, from here, the only moves lead to dead ends so there is not anything useful for the bubble to do besides immediately return to the cycle. There are no other legal moves while turning the second cycle. While rotating the third cycle, it is possible for the bubble to reach the shared vertices of the leftmost edge gadget, but by Lemma~\ref{no edge leakage} this does not help. While rotating the fourth cycle, it is possible for the bubble to reach the shared vertices of the rightmost edge gadget, but again this does not help.

  Now we consider when the bubble enters the OR vertex gadget from an edge gadget through one of the shared vertices. In the first state, there are no legal moves after entering from the top or right edges. In the second state, from either the top or left edges it can enter and traverse most of the gadget but cannot complete any loop and thus cannot make progress by Lemma~\ref{no loops}. In the third state, the bubble has no legal moves after entering from the left or top edge. From the right edge it can traverse most of the gadget but cannot complete any loops. From the fourth state, again, while the bubble can traverse most of the gadget after entering from the top edge, it does not complete any loops so it has no effect. In the fifth state, there are no legal moves after entering the vertex gadget.
  
  \begin{lemma}
    \label{no loops}
    If the bubble takes any path from any vertex to the same vertex which does not complete a nontrivial loop, then the state of the Subway Shuffle instance must not have changed.
  \end{lemma}
  \begin{proof}
    If the bubble never completed a loop, then the only way for it to get back to where it started is to take the same path in reverse. By the definition of Subway Shuffle moves, this exactly undoes these moves returning the instance back to its original state.
    \end{proof}

  Finally, we check the win gadget. While using the win gadget, there are no legal moves other than completing the one loop. There is only one edge connected to the win gadget. If the bubble tries to enter the win gadget here, it cannot leave anywhere else or complete any loops, so by Lemma~\ref{no loops} it must return with no effect.

  Since the constraint logic graph is planar, the reduction yields a planar graph for 2-color oriented Subway Shuffle. Since the constructed instance Subway Shuffle is winnable exactly when the constraint logic instance is, and the reduction can clearly be done in polynomial time, this shows 2-color oriented Subway Shuffle is \PSPACE-hard. All of the gadgets used, including the trees connected gadget entrances, use only valid vertices, so it is still \PSPACE-hard with only valid vertices.
\end{proof}

\section{\onerush{} is \PSPACE-complete}\label{sec:rush hour PSPACE}

In this section, we show that \onerush{} is \PSPACE-complete by a reduction from 2-color oriented Subway Shuffle with only valid vertices and only a single empty vertex, which was shown to be \PSPACE-complete in the previous section. \onerush{} is played on a large square grid. We allow for fixed blocks, which are spaces marked impassable in the grid.

We will simulate Subway Shuffle vertices with individual cars at intersections, and edges as paths of cars. In general, purple edges and vertices will be horizontal cars, and orange edges and vertices will be vertical cars. Like in the Subway Shuffle, we will have a single \emph{bubble} which is a single empty space that moves around as cars move into that space.

We replace each vertex in our Subway Shuffle instance with a single car which is vertical if there is an orange token there, and horizontal if a purple token is there. Orange edges leading from a vertex attach to it as vertical rows of cars, and purple edges attach to a vertex as horizontal rows of cars. A degree-4 vertex with a purple token is depicted in Figure~\ref{fig:vertex}. Valid vertices can be embedded this way, with fixed blocks on the unused sides for lower degree vertices.

\begin{figure}
  \centering
  \begin{subfigure}{.3\linewidth}
    \centering
    \includegraphics[width=\linewidth]{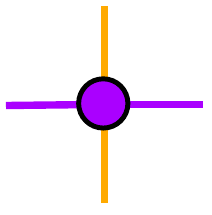}
    \label{subfig:ss vertex}
  \end{subfigure}
  \hfil
  \begin{subfigure}{.3\linewidth}
    \centering
    \includegraphics[width=\linewidth]{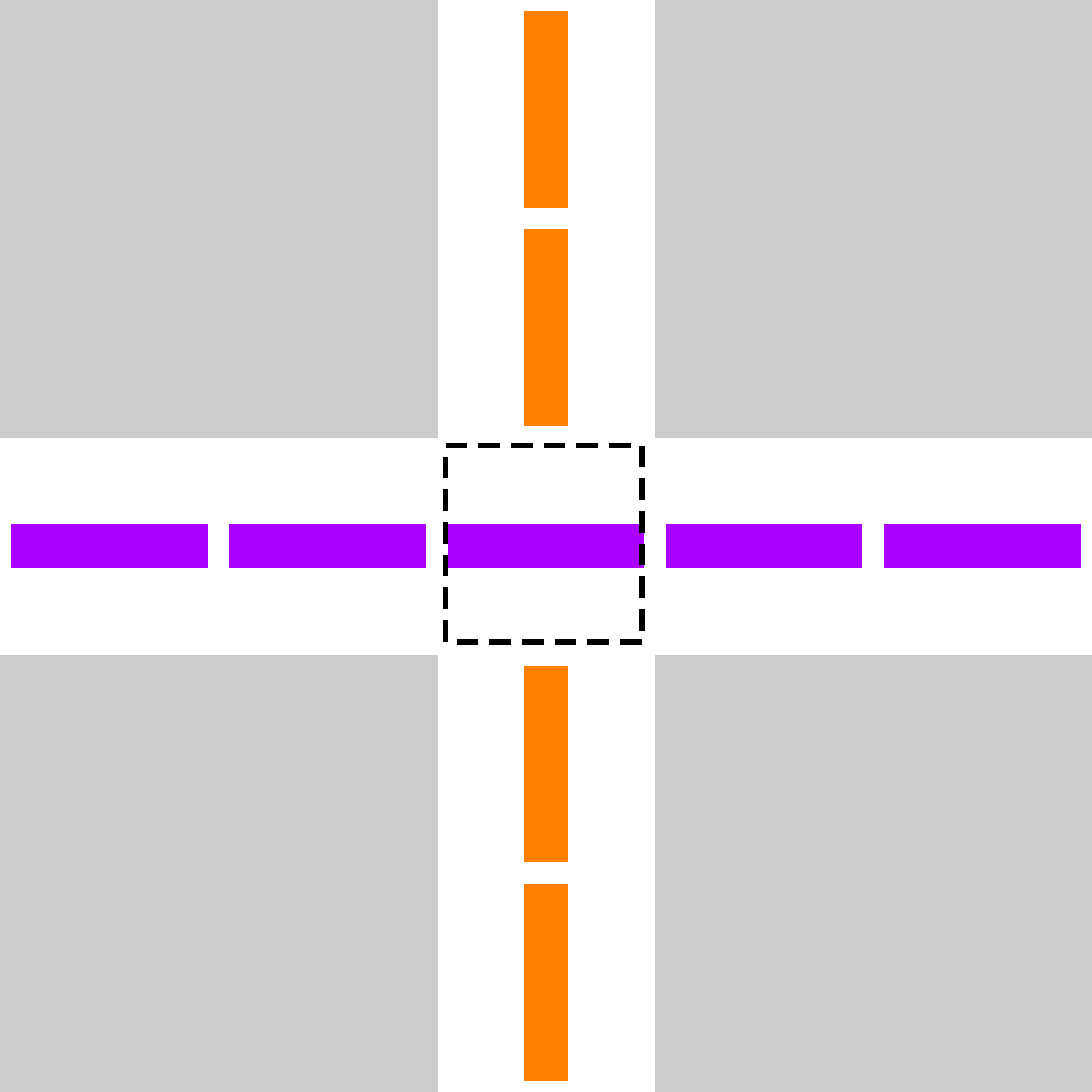}
    \label{subfig:rh vertex}
  \end{subfigure}
  \caption{A degree-4 Subway Shuffle vertex embedded in Rush Hour. Note that, while this is not a valid Subway Shuffle vertex, all valid vertices are subsets of this vertex. Individual dashes represent cars. A line of cars of one color represents a Subway Shuffle edge of that color. The center boxed car represents the Subway Shuffle vertex.}
  \label{fig:vertex}
\end{figure}

A Subway Shuffle edge is simulated by a path of cars which can make right-angle turns, allowing us to embed an arbitrary planar Subway Shuffle graph.  The direction of a car at a turn in an edge defines which way the Subway Shuffle edge is oriented. A purple edge which points right is depicted in Figure~\ref{fig:wire}. In order to maintain the directionality of edges, each edge must be simulated by a path with at least one turn.
\begin{figure}
  \centering
  \begin{subfigure}{.3\linewidth}
    \centering
    \includegraphics[width=\linewidth]{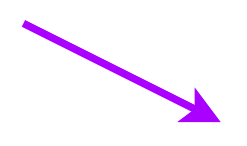}
    \label{subfig:ss wire}
  \end{subfigure}
  \hfil
  \begin{subfigure}{.3\linewidth}
    \centering
    \includegraphics[width=\linewidth]{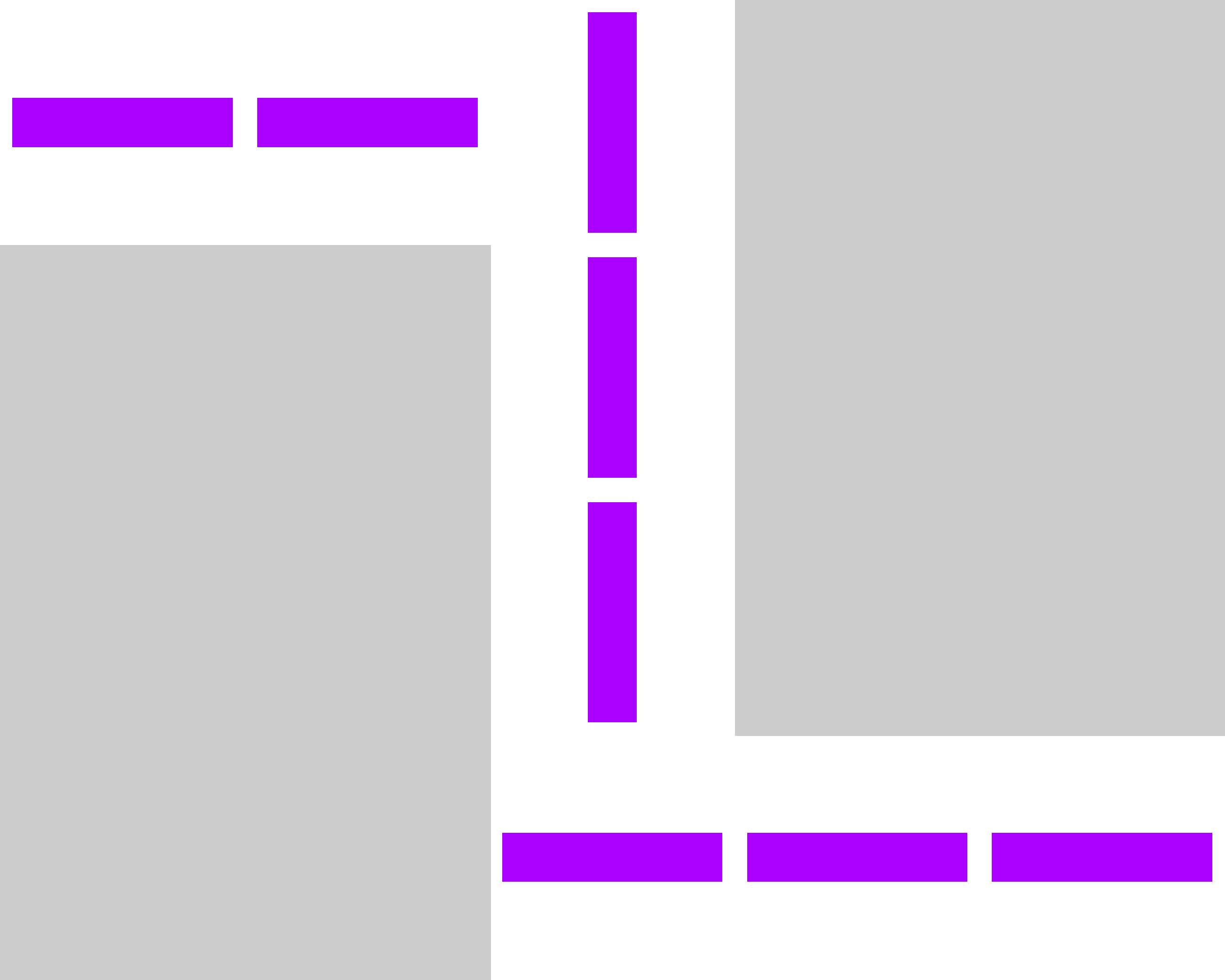}
    \label{subfig:rh wire}
  \end{subfigure}
  \caption{A Rush Hour simulation of a Subway Shuffle edge. This is a purple edge which points right.}
  \label{fig:wire}
\end{figure}

To make a move, suppose the bubble is currently at a vertex. To move a token in from an adjacent vertex, a car from the connecting edge is moved in. Then cars from that edge are all moved one space toward the initial vertex, until finally we can move the car in the second vertex out. Note that this process reverses the orientation of the edge as desired. If the edge was pointed in the correct direction, then this process will succeed; if the edge is oriented in the wrong direction, then this process will fail when we try to turn a corner in the edge. Similarly it is impossible to move a token along an edge of the opposite color, because it will be unable to move out of its vertex. An example of a single Subway Shuffle move where an orange token is moved up along an orange edge embedded in Rush Hour is shown in Figure~\ref{fig:rush_hour_move}.

No other useful actions can be taken. If the bubble is not currently at a vertex, then there are at most two possible moves. One of them would just be undoing the previous move, and the other would be continuing the process of moving a token along an edge. When the bubble is at a vertex, moving any adjacent car into the vertex is the same as starting the process of moving a Subway Shuffle token along the corresponding edge.

\begin{figure}
  \centering
  \begin{subfigure}{.3\linewidth}
    \centering
    \includegraphics[width=\linewidth]{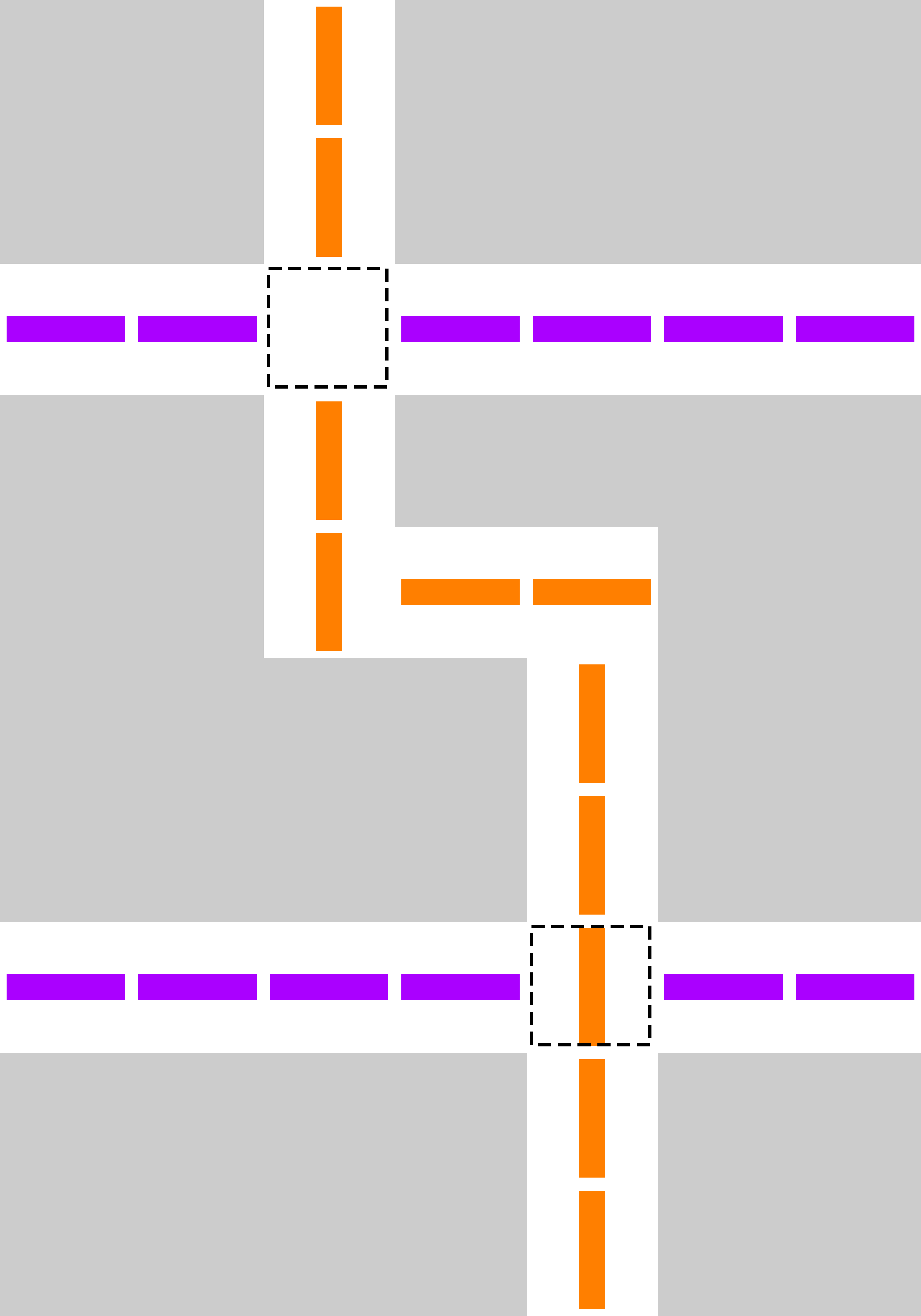}
    \caption{Before moving the orange token up.}
    \label{subfig:move_before}
  \end{subfigure}
  \hfil
  \begin{subfigure}{.3\linewidth}
    \centering
    \includegraphics[width=\linewidth]{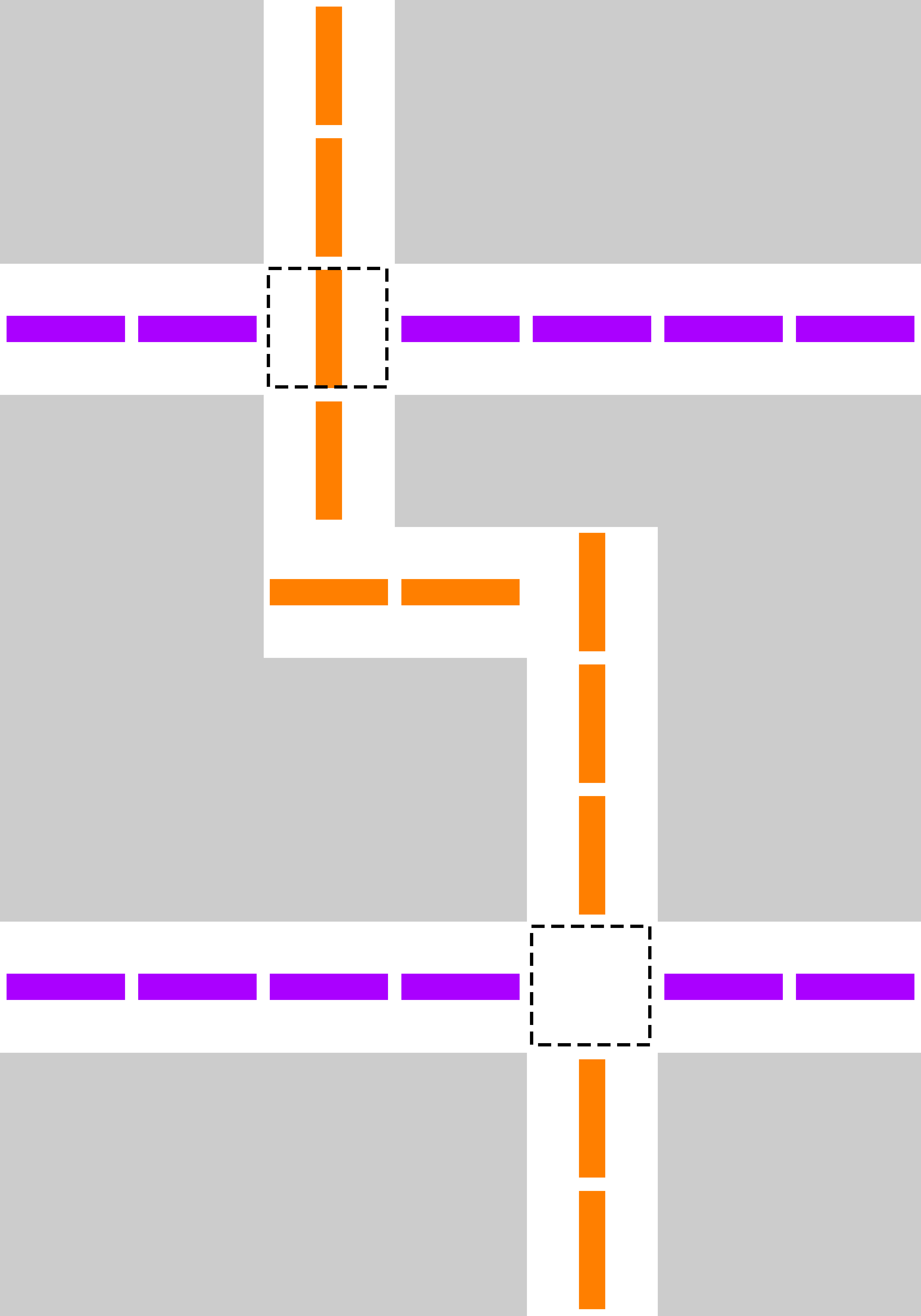}
    \caption{After moving the orange token up.}
    \label{subfig:move_after}
  \end{subfigure}
  \caption{Moving a single orange token in Subway Shuffle
           when simulated by Rush Hour in Figure~\ref{fig:wire}.}
  \label{fig:rush_hour_move}
\end{figure}

The win condition of a Rush Hour instance is allowing the marked car to escape the grid. The win gadget needs to be specified more precisely because Subway Shuffle tokens do not correspond exactly to Rush Hour cars. Also, we want to make sure that everything can fit within a grid so our win condition is actually located near the edge.
\begin{figure}
  \centering
  \includegraphics[width=.3\linewidth]{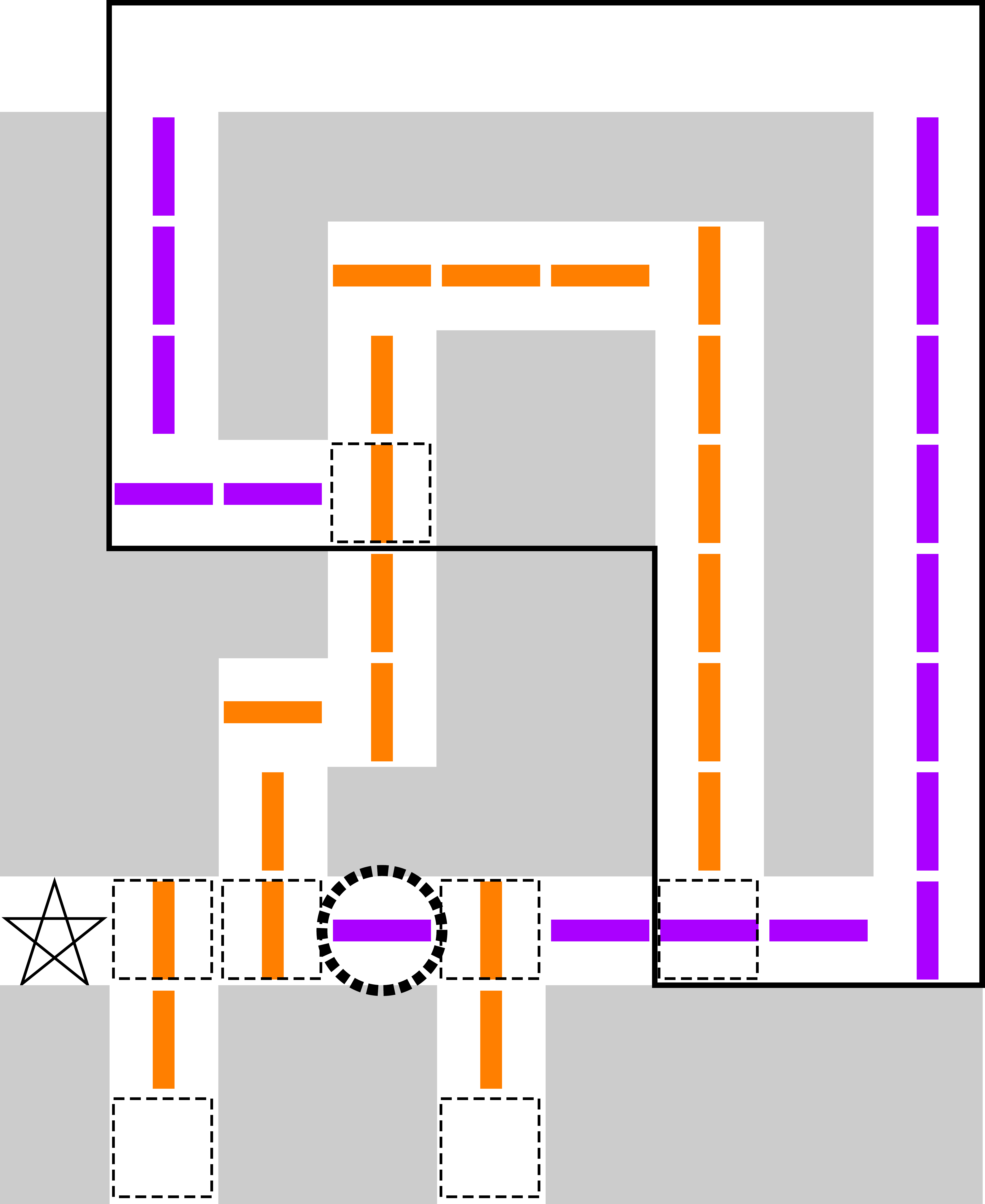}
  \caption{The cycle of the subway shuffle win gadget embedded in Rush Hour. The goal is to get the circled car to the star. The boxed cars are the vertices in the Subway Shuffle win gadget. The two lines of purple cars extending upward are the purple edges of the connected edge gadget. Everything inside the solid black line is part of the connecting edge gadget.}
  \label{fig:win}
\end{figure}

Our win gadget is depicted in Figure~\ref{fig:win}. The win condition is the circled car reaching the star. The boxed cars represent Subway Shuffle vertices. In order to win, first the boxed orange car directly in front of the circle car must leave by rotating this cycle. This represents the marked token in the Subway Shuffle vertex moving to the middle vertex along the bottom of the win gadget. Then, the leftmost orange line must be moved down one space, clearing the way for the marked car to leave.

In Rush Hour, because winning requires a car leaving the grid, we must also take care to make sure that the win gadget is at the boudary of our construction, and not somewhere buried in the middle. To do this, we conisder the CL edge which is part of the win gadget. Since the CL graph is planar, we can consider one of the faces that this edge is a part of, and make this face the ``outside'' face. Now our win gadget is at the boundary, which is what we needed.

\section{Open Problems} \label{sec:conclusion}

In this paper, we have shown that {\onerush} with fixed blocks is
\PSPACE-complete, solving Tromp and Cilibrasi's open problem \cite{TC}.
It remains whether the assumption of fixed blocks can be eliminated,
and thereby solve the open problem of Hearn, Demaine, and Tromp \cite{HD05,TC}. We note that it is impossible to perfectly simulate a fixed block using Rush Hour cars, since for any arrangement of cars in a region, there must be at least one point along the boundary of the region that, if it were empty, a car can exit the region. For a single bubble, it gets worse than that. Let a space be \emph{accessible} if the bubble can ever reach that space. By Theorem~\ref{bubble access}, the accessible region is always a rectangle. Since we can ignore anything inaccessible, we can just assume that everywhere in the entire Rush Hour grid is accessible. Because the bubble can get everywhere, it seems impossible to modify the gadgets in our proof in any simple way to constrain the bubble from wandering freely inside and between the cycles in gadgets.

  \begin{theorem}
    \label{bubble access}
  In any {\onerush} instance with no fixed blocks with only a single ``bubble,'' the set of accessible spaces is a rectangle.
\end{theorem}
\begin{proof}
  The accessible region is clearly connected. If it is not a rectangle, there must be a corner on the boundary of the accessible region where two accessible spaces are adjacent to the same inaccessible space, as in Figure~\ref{fig:corner access}. Then regardless of its orientation, the car in this inaccessible space must be able to move into one of these two accessible spaces, and thus is also accessible. This is a contradiction, so the accessible region must be a rectangle.
\end{proof}

\begin{figure}
  \centering
  \begin{subfigure}{.2\linewidth}
    \includegraphics[width=.99\linewidth]{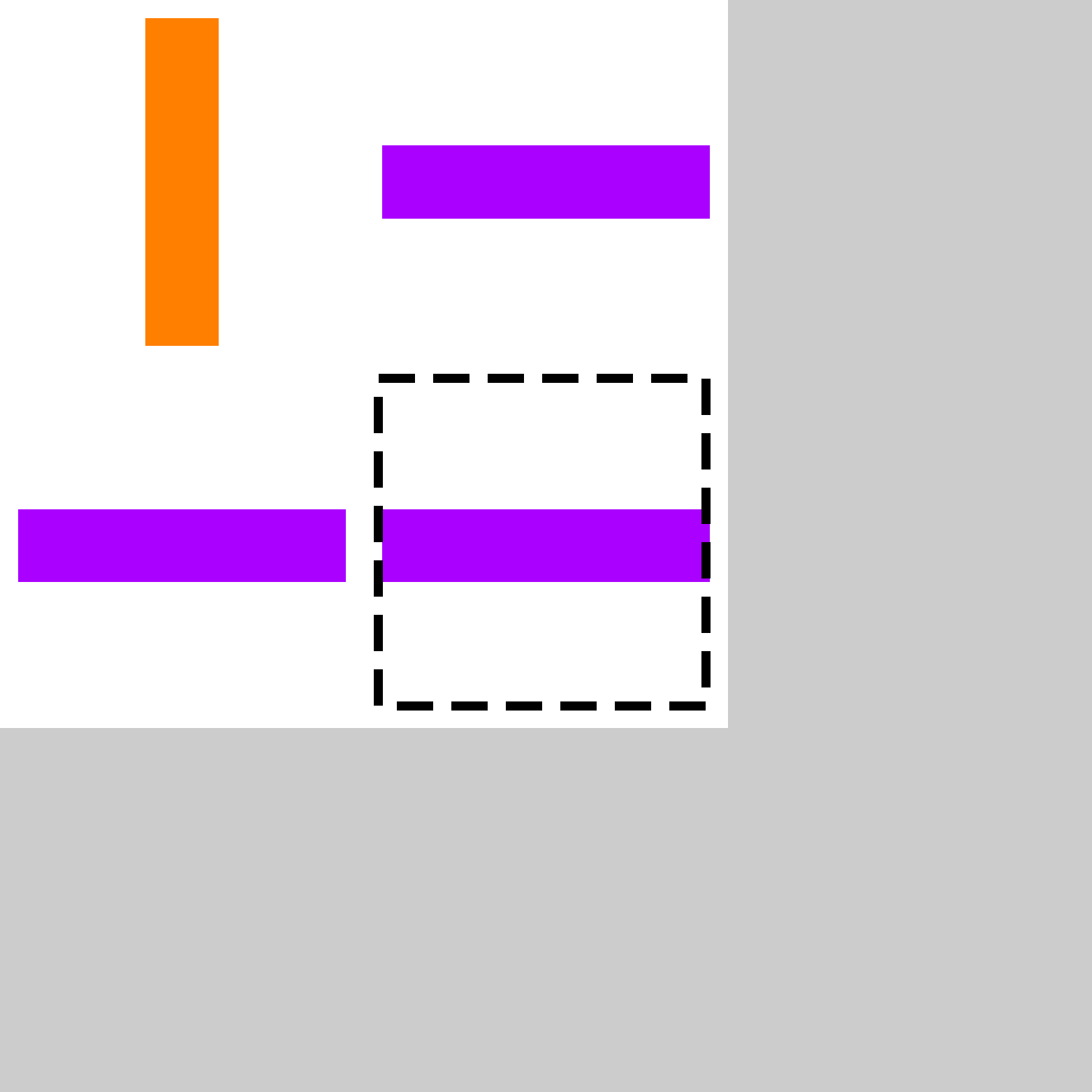}

    \end{subfigure}
  \caption{Let the gray area be accessible by the bubble. Then the boxed car is at the corner of the boundary of the accessible region, and regardless of its orientation it must also be accessible by the the bubble.}
  \label{fig:corner access}
\end{figure}
   



\section*{Acknowledgments}

We thank the many colleagues over the years for their early collaborations in
trying to resolve the $1 \times 1$ Rush Hour problem
(when E. Demaine mentioned it to various groups over the years):
Timothy Abbott, Kunal Agrawal, Reid Barton, Punyashloka Biswal, Cy Chen, Martin Demaine, Jeremy Fineman, Seth Gilbert, David Glasser, Flena Guisoresac, MohammadTaghi Hajiaghayi, Nick Harvey, Takehiro Ito, Tali Kaufman, Charles Leiserson, Petar Maymounkov, Joseph Mitchell, Edya Ladan Mozes, Krzysztof Onak, Mihai P\v{a}tra\c{s}cu, Guy Rothblum, Diane Souvaine, Grant Wang, Oren Weimann, Zhong You
(MIT, November 2005);
Jeffrey Bosboom, Sarah Eisenstat, Jayson Lynch, and Mikhail Rudoy (MIT 6.890, Fall 2014);
and
Joshua Ani, Erick Friis, Jonathan Gabor, Josh Gruenstein, Linus Hamilton, Lior Hirschfeld, Jayson Lynch, John Strang, Julian Wellman
(MIT 6.892, Spring 2019, together with the present authors).

\bibliographystyle{alpha}
\bibliography{paper}
\end{document}